\documentclass[11pt]{article}

\usepackage{float}
\usepackage{graphicx,color}
\usepackage{epsfig}
\usepackage{amsmath,amssymb,amsthm}
\usepackage{subfigure}
\usepackage{dsfont}
\usepackage{psfrag}
\usepackage{algpseudocode}

\def\1{\mathbf{1}}

\DeclareMathOperator{\erf}{erf}

\graphicspath{{figures/}}

\floatstyle{ruled}
\newfloat{Algorithm}{tbp}{lop}[section]

\newtheorem{thm}{Theorem}

\begin{document}

\title{Uncertainty Quantification in Hybrid Dynamical Systems}

\author{Tuhin Sahai$^{\dagger}$ \and Jos{\'e} Miguel Pasini$^{\dagger}$}
\date{$^\dagger$United Technologies Research Center, 411 Silver Lane, East Hartford, CT 06108, USA}

\maketitle
\begin{abstract}
Uncertainty quantification (UQ) techniques are frequently used
to ascertain output variability in systems with parametric
uncertainty. Traditional algorithms for UQ are either
system-agnostic and slow (such as Monte Carlo) or fast with
stringent assumptions on smoothness (such as polynomial chaos and
Quasi-Monte Carlo). In this work, we develop a fast UQ approach
for hybrid dynamical systems by extending the polynomial chaos
methodology to these systems. To capture discontinuities, we use a wavelet-based Wiener-Haar expansion. We develop a boundary layer approach to propagate uncertainty through \emph{separable} reset conditions. We also introduce a transport theory based approach for propagating uncertainty through hybrid dynamical systems. Here the expansion yields a set of hyperbolic equations that are solved by integrating along characteristics. The solution of the partial differential equation along the characteristics allows one to quantify uncertainty in hybrid or switching dynamical systems. The above methods are demonstrated on example problems.

\end{abstract}


\section{Introduction}


Uncertainty Quantification (UQ) is an area of mathematics that
is used to quantify output distributions given parametric
uncertainty. Traditional approaches include Monte Carlo and
Quasi-Monte Carlo methods~\cite{McQMc}, response surface
methods~\cite{Response1:book,Response2:book} as well as polynomial
chaos and probabilistic collocation based
approaches~\cite{PolyReview}. The polynomial chaos approach for
uncertainty quantification was originally proposed by Norbert
Wiener~\cite{Wiener}. Assuming that one is given input
uncertainty in the form of distributions associated with various
parameters of the system, polynomial chaos/probabilistic
collocation methods provide an approach for fast uncertainty
quantification under the assumption of smooth dynamics. In
particular, polynomial chaos provides exponential convergence
for smooth systems and processes with finite
variance~\cite{PolyReview}. Polynomial chaos based methods have
been used for a multitude of
applications, see~\cite{Allen2009,Elman2011,Ghanem1998,
Najm2009,TuhinPoly,Xiu2003,Zabaras2008,Multigpc} for examples. Note that, depending on
the application, one can combine various UQ
approaches. For example, a combination of polynomial chaos and the response surface
methodology has been used to develop
probabilistic collocation methods for discrete distributions in~\cite{TuhinPoly}.

In this work, we focus on developing UQ techniques for hybrid dynamical systems.
Hybrid dynamical systems theory is used to model systems with both
discrete and continuous dynamics~\cite{SHS:book}. Examples
include the bouncing ball automaton~\cite{ZenoReg}, biological
networks~\cite{HybridCompCont:book,BioEx1}, air traffic
management systems~\cite{AirTraffic}, communication
networks~\cite{CommNet}, elevators, and robotics, to name a few.
These systems frequently display rich dynamics not seen in continuous
systems. For example, Zeno behavior in hybrid systems is characterized by an infinite number of discrete switches in
finite time~\cite{ZenoReg,Zeno}. Hybrid systems can be particularly challenging from an analysis standpoint since traditional techniques, such as polynomial chaos based methods, assume smoothness, rendering them inapplicable.

In this work, we develop polynomial chaos and transport theory based methods for
propagating uncertainty through hybrid systems. We assume that the domains associated with different modes of operation of the hybrid system do not overlap. We demonstrate that, by integrating over
appropriate time-varying regions, one can extend the polynomial
chaos framework to hybrid dynamical systems. We resolve the
issue of state resets~\cite{SHS:book} in the separable case by using boundary layer
approximations. To capture the discontinuities in the probability distributions of the output variables, we use a Haar-wavelet expansion~\cite{Haar1910}. This expansion has previously been used in the polynomial chaos setting to propagate uncertainty through dynamical systems close to bifurcation points~\cite{Najm2009,LeMaitre2004,LeMaitre2004a}. Here we develop a methodology to propagate uncertainty through systems with discontinuities in dynamics and output along with state resets. We also develop a transport theory based approach that allows one to propagate the uncertainties through the various modes of the hybrid dynamical system.

Our paper is organized as follows: in section~\ref{Sec:problem} we define hybrid dynamical systems and the problem of uncertainty quantification. In section~\ref{Sec:hyb_poly} we first construct the framework for polynomial chaos in the hybrid dynamical system setting~(\ref{Sec:hpc}). We then demonstrate the Haar wavelet expansion for hybrid polynomial chaos in~\ref{Sec:wavelets}. The handling of state resets is considered in section~\ref{Sec:resets}. Finally, the results on hybrid polynomial chaos are presented in~\ref{Sec:Results}. The transport operator theory based method for propagating uncertainty through hybrid dynamical systems is developed in section~\ref{sec:transporttheory} and conclusions are drawn in section~\ref{Sec:conclusions}.

\section{Problem definition}\label{Sec:problem}

Let  $S=(q, X, f, x(0), D, E, G, R)$ denote a hybrid system $S$, where \\
\centerline{
\begin{tabular}{l l}
$q$ & Set of discrete variables \\
$X$ & Set of continuous variables\\
$f:X\times q \rightarrow TX$& Vector field\\
$x(0)$ & Set of initial conditions\\
$D:q\rightarrow P(X)$ & Domain\\
$E$ & Set of discrete transitions\\
$G: E\rightarrow P(X)$ & Guard conditions\\
$R: E\times X \rightarrow P(X)$ & Reset map.\\
\end{tabular}}
In the above table, $TX$ is the tangent bundle of $X$ and $P(X)$
is the power set of $X$. For more details on the definition of
hybrid systems, see~\cite{SHS:book}.

We can use the following representation for hybrid systems,
\begin{equation}
\dot x = f(x,\lambda,q), \label{eq:fullsys}
\end{equation}
where $x \in X$ is a vector of state variables and the form of
$f(x,\lambda,q)$ is dictated by~$q$, which represents the mode
of operation of the hybrid dynamical system. The discrete state~$q$ is
determined by the guard conditions that dictate transitions
between modes (see Fig.~\ref{Fig:HySchematic}). The reset
functions $h_{i}(x)$ are a part of the reset map~$R$. In
Eqn.~\ref{eq:fullsys} let~$\lambda$ denote the vector of system
parameters and $x(0)$ the initial condition for the
system.

\begin{figure}
  \centering
  \includegraphics[width=0.7\hsize]{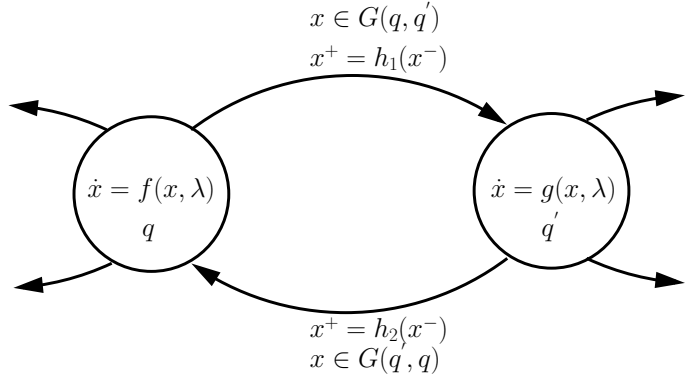}\\
\caption{Schematic for hybrid (switching) systems.
\label{Fig:HySchematic} }
\end{figure}

If the system parameters~$\lambda$ in Eqn.~\ref{eq:fullsys} are
uncertain (i.e., each $\lambda_{i}$ has an associated
distribution) then one typically desires to quantify the
time-varying moments (such as mean and variance) of~$x(t)$ (note that $x(0)$ may also be uncertain). As
mentioned earlier, although one can use Monte Carlo based
sampling methods~\cite{McQMc}, they are plagued by slow
convergence. In particular, the mean is expected to converge as
$1/\sqrt{N}$, where $N$ is the number of samples. Quasi-Monte
Carlo based sampling methods are expected to give a convergence
rate of $\log^d (N)/N$, where $d$ is the dimensionality of the
random space~\cite{QMC}, making these methods attractive for
problems in low dimensions. Polynomial chaos based methods
provide an alternative framework for uncertainty quantification
with exponential convergence for processes with finite variance,
but they too suffer from the curse of
dimensionality~\cite{PolyReview}. In the next section we extend
the polynomial chaos framework to hybrid systems.

\section{Polynomial chaos for hybrid dynamical systems}\label{Sec:hyb_poly}

Starting with a complete probability space $\Gamma$ given by
$(\Omega, \mathcal{F},\mathbb{P})$, where $\Omega$ is the sample
space, $\mathcal{F}$ is the $\sigma$-algebra on $\Omega$ and
$\mathbb{P}$ is a probability measure, let $L_{2}(\Gamma,X)$
denote the Hilbert space of square-integrable,
$\Gamma$-measurable, $X$-valued random elements. Then one can,
in general, define a polynomial chaos basis
$\{H_{\alpha}(\lambda(\omega))\}$, where $\lambda(\omega)$ is a
random vector and $\alpha = (\alpha_{1},\alpha_{2},\dots)$ is a
vector of non-negative indices. We denote the probability
density function of the random vector $\lambda$ by
$\rho(\lambda)$.

Generalized polynomial chaos (gPC)~\cite{BeyondWienerAskey}, provides a framework for
representing second-order stochastic processes $r\in
L_{2}(\Gamma,X)$ for arbitrary distributions of $\lambda$ by the following expansion:
\begin{equation}
r(\lambda) =
\displaystyle\sum_{|\alpha|=0}^{\infty}a_{\alpha}H_{\alpha}(\lambda),
\label{eq:expan1}
\end{equation}
where $|\alpha| = \sum_{i} \alpha_{i}$ is the sum of the indices of $\alpha$ and
$H_{\alpha}(\lambda)$ are orthogonal polynomials on $\Gamma$
with respect to $\rho(\lambda)$, i.e.
\begin{equation}
\displaystyle\int_{\Gamma}
\rho(\lambda)H_{\alpha}(\lambda)H_{\beta}(\lambda)d\lambda =
\delta_{\alpha\beta}, \label{eq:ortho}
\end{equation}
where $\delta_{\alpha\beta}$ is the Kronecker delta product.
Depending on $\rho(\lambda)$ one can generate an appropriate
orthogonal basis for representing~$r(\lambda)$. For example, if $\rho$ is
Gaussian, then the appropriate polynomial chaos basis is the set
of Hermite polynomials; if $\rho$ is the uniform distribution,
then the basis is the set of Legendre polynomials. For details
on the correspondence between distributions and polynomials
see~\cite{PolyReview,Ogura}. A framework to generate polynomials
for arbitrary distributions has been developed
in~\cite{BeyondWienerAskey}.

In practice, the expansion in Eqn.~\ref{eq:expan1} is truncated
at a particular order, say, $p$. One can then use Galerkin projections to
obtain a set of differential equations for the coefficients
$a_{\alpha}$ in Eqn.~\ref{eq:expan1}~\cite{PolyReview}.

We now extend the standard polynomial chaos framework to hybrid
dynamical systems despite the presence of switching and state resets. To the best of our knowledge, it is the first attempt
to develop tools for fast uncertainty quantification for this class of systems.

\subsection{Hybrid Polynomial Chaos}
\label{Sec:hpc}

Without loss of generality, consider the following two-mode
hybrid dynamical system as representative of systems in which the different operating modes are associated with non-overlapping regions:
\begin{equation}
\dot{x} = \begin{cases} f(x,\lambda) & \text{if $x\geq 0$} \\
g(x,\lambda) & \text{otherwise}.
\end{cases}
\label{eq:wlogsys}
\end{equation}
Here one desires to quantify $x(t;\lambda)$, i.e., determine
$x$ as a function of time $t$ and parameters $\lambda$. The system above has two
modes of operation determined by its state. One
can parameterize these modes in the following way:
\begin{eqnarray}
\1_{R_{1}}(x) &=& \begin{cases}
1 & \text{if $x\geq 0$} \\
0 & \text{otherwise}
\end{cases}
 \label{eq:indict_r1} \\
\1_{R_{2}}(x) &=& 1 - \1_{R_{1}}(x). \label{eq:indict_r2}
\end{eqnarray}
When $\1_{R_{1}}(x) = 1$ (corresponding to $x\geq 0$) the governing
differential equations are $f(x,\lambda)$, and when $\1_{R_{2}}(x)=1$
($x < 0$) the governing differential equations are
$g(x,\lambda)$. Thus, one can rewrite Eqn.~\ref{eq:wlogsys} as,
\begin{equation}
\dot x = \1_{R_{1}}(x) f(x,\lambda) + \1_{R_{2}}(x) g(x,\lambda).
\label{eq:homotop}
\end{equation}
This equation extends easily to $k$ modes of operation by constructing indicator functions for each mode of operation $\1_{R_{1}},\1_{R_{2}},\hdots,\1_{R_{k}}$ of the hybrid system. We now expand $x(t;\lambda)$ in the appropriate orthogonal
polynomial chaos basis,
\begin{equation}
x(t;\lambda) =
\displaystyle\sum_{|\alpha|=0}^{p}a_{\alpha}(t)H_{\alpha}(\lambda).
\label{eq:expx}
\end{equation}
Dropping the arguments of $a_{\alpha}(t)$ and
$H_{\alpha}(\lambda)$ for simplicity and using the above
relation with Eqn.~\ref{eq:homotop}, one gets
\begin{eqnarray}
\sum_{|\alpha|=0}^{p}\dot a_{\alpha} H_{\alpha} &=&
\1_{R_{1}}(\sum_{|\alpha|=0}^{p}a_{\alpha}H_{\alpha})f(\sum_{|\alpha|=0}^{p}a_{\alpha}H_{\alpha},\lambda)
\nonumber\\
&+& \1_{R_{2}}(\sum_{|\alpha|=0}^{p}a_{\alpha}H_{\alpha})
g(\sum_{|\alpha|=0}^{p}a_{\alpha}H_{\alpha},\lambda). \nonumber
\label{eq:exphomotop}
\end{eqnarray}
By multiplying the above relation by
$\rho(\lambda)H_{k}(\lambda)$, integrating over $\Gamma$, and
using orthogonality conditions, we get
\begin{eqnarray}
\dot a_{k}(t) &=&
\int_{R_{1}(t)}f(\sum_{|\alpha|=0}^{p}a_{\alpha}H_{\alpha},\lambda)\rho(\lambda)H_{k}(\lambda)d\lambda
\nonumber\\ &+&
\int_{R_{2}(t)}g(\sum_{|\alpha|=0}^{p}a_{\alpha}H_{\alpha},\lambda)\rho(\lambda)H_{k}(\lambda)d\lambda.
\label{eq:intregionpc}
\end{eqnarray}
Note that $R_{1}(t) = \{\lambda:
\sum_{|\alpha|=0}^{p}a_{\alpha}H_{\alpha} \geq 0\}$ and
$R_{2}(t) = \Gamma - R_{1}(t)$. Thus by evaluating the two
integrals one can evolve $a_{k}(t)$ for any index vector~$k$.
Note, however, that the regions of integration $R_{1}$
and~$R_{2}$ are time-dependent quantities and must be evaluated
at every instant in time. 
For a pictorial depiction of $R_{1}$ and $R_{2}$, see Fig.~\ref{Fig:Region}.

\begin{figure}
  \centering
  \psfrag{a}[][]{$\sum_{|\alpha|=0}^{p}a_{\alpha}(t)H_{\alpha}(\lambda)$}
  \psfrag{b}[][]{$R_{1}(t)$}
  \psfrag{c}[][]{$R_{2}(t)$}
  \includegraphics[width=0.7\hsize]{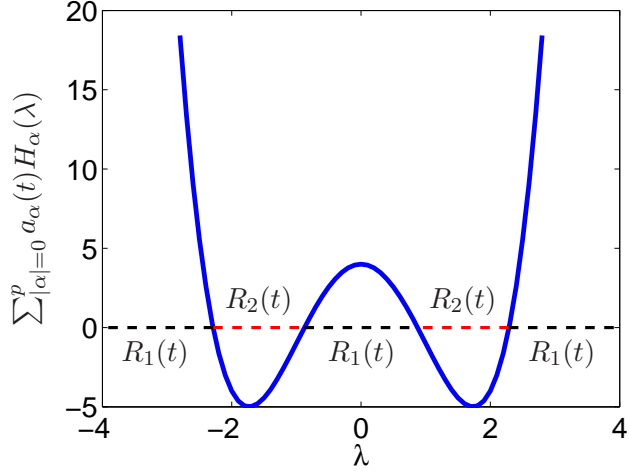}\\
  \caption{Pictorial representation of the regions of integration for the two integrals in
  Eqn.~\ref{eq:intregionpc}. In this example, the regions are not simply connected.
  \label{Fig:Region}}
 \end{figure}

\subsection{Hybrid Polynomial Chaos and wavelet expansions}
\label{Sec:wavelets}

Hybrid systems can display discontinuous behavior as a function of the uncertain parameters. In view of this, a smooth polynomial chaos expansion is expected to degrade as the discontinuities become more severe. In Ref.~\cite{LeMaitre2004} the authors develop a wavelet-based Wiener-Haar expansion to treat bifurcating (but smooth) dynamical systems with uncertain initial conditions that result in discontinuous behavior. In this section we adapt the Wiener-Haar expansion to hybrid dynamical systems.

In~\cite{LeMaitre2004, LeMaitre2004a}, output variables are expanded in terms of Wiener-Haar wavelets expressed as functions of the Cumulative Distribution Function (CDF) of the uncertain parameters. For simplicity, consider the univariate case. Here we denote the CDF of the uncertain parameter~$\lambda$ as $u(\lambda)$ and expand the state vector~$x$ as\footnote{For the multivariate case, see Ref.~\cite{LeMaitre2004}.}
\begin{equation}
x(t; \lambda) = x_0(t) + \sum_{j=0}^P \sum_{k=0}^{2^j-1} x_{jk}(t) \psi_{jk} (u(\lambda)),
\label{eq:x_waveletexp}
\end{equation}
where,
\[
\psi_{jk}(u) = 2^{j/2} \psi(2^j u - k), \quad \text{(with $j=0,1,\ldots$ and $k=0,\ldots,2^j-1$)}
\]
is a family of Haar wavelets~\cite{Haar1910}, defined in terms of the \emph{mother wavelet}:
\[
\psi(u) = \begin{cases}
1  & 0 \leq u < 1/2 \\
-1 & 1/2 \leq u < 1 \\
0 & \text{otherwise.}
\end{cases}
\]
The index~$j$ determines the scale of the wavelet and $k$ its displacement. Note that $\{\psi_{jk}\}$ is a family of orthonormal functions on the interval $[0,1]$ with respect to the uniform density. This makes the family $\{\psi_{jk} \circ u\}$ automatically orthonormal with respect to the probability density of~$\lambda$:
\[
\delta_{jl} \delta_{km} = \int_0^1 \psi_{jk}(u) \psi_{lm}(u) du
= \int (\psi_{jk} \circ u)(\lambda) \, (\psi_{lm} \circ u) (\lambda) \rho(\lambda) d\lambda.
\]
Additionally, all $\psi_{jk}$'s are orthogonal to the constant function on $[0,1]$, which implies that the mean of $x$ is given by the first term in the expansion
\[
x_0(t) = \int x(t;\lambda) \rho(\lambda) d\lambda
\]
and that the variance is
\[
\sigma^2(t) = \sum_{j=0}^P \sum_{k=0}^{2^j-1} x_{jk}^2(t).
\]

We now use this expansion on a switching oscillator example:
\begin{align}
\ddot x + c\dot x + x + \lambda &= 0\,\,\,\text{if $x \geq 0$} \nonumber \\
\ddot x + c\dot x + x - \lambda &= 0\,\,\,\text{if $x < 0$},
 \label{eq:SHMswitch}
\end{align}
which can be rewritten as,
\begin{align*}
\dot{x} &= y \\
\dot{y} &= - cy -x - \lambda \mathbf{1}_{R_1}(x) + \lambda \mathbf{1}_{R_2}(x).
\end{align*}
The expansion in this case is
\begin{align*}
x(t;\lambda) &= x_0(t) + \sum_{j=0}^P \sum_{k=0}^{2^j-1} x_{jk}(t) \psi_{jk} (u(\lambda)) \\
y(t;\lambda) &= y_0(t) + \sum_{j=0}^P \sum_{k=0}^{2^j-1} y_{jk}(t) \psi_{jk} (u(\lambda)).
\end{align*}
Projecting these equations onto the basis functions yields,
\begin{align*}
\dot{x}_0 &= y_0 \\
\dot{y}_0 &= -c y_0 - x_0  - \int_0^1 \lambda(u) \mathbf{1}_{R_1}(x(u)) du + \int_0^1 \lambda(u) \mathbf{1}_{R_2}(x(u)) du \\
\dot{x}_{jk} &= y_{jk} \\
\dot{y}_{jk} &= -c y_{jk} - x_{jk}  - \int_0^1 \lambda(u) \mathbf{1}_{R_1}(x(u)) \psi_{jk}(u) du
+ \int_0^1 \lambda(u) \mathbf{1}_{R_2}(x(u)) \psi_{jk}(u) du
\end{align*}
Note that to compute $\lambda(u)$ we invert the CDF~$u(\lambda)$.

To compute the integrals needed to evolve these equations numerically, we take advantage of the fact that Haar wavelets are piecewise constant. Namely, for a given truncation order~$P$, if we divide $[0,1]$ into $2^{P+1}$ equal subintervals, both $\psi_{jk}$ (with $j \leq P$) and the truncated expansion for $x$ (and therefore the indicator functions) are constant in each subinterval. This implies that in each subinterval we only need to calculate the integral of $\lambda(u)$, which is known a priori. For the case of a Gaussian $\lambda \sim N(\mu,\sigma^2)$, we have
\[
\lambda(u) = \mu + \sigma \sqrt{2} \erf^{-1} (2u-1),
\]
which has a primitive,
\[
\int \lambda(u) du = \mu u - \sigma
\frac{1}{\sqrt{2\pi}} \exp\left\{ -[\erf^{-1}(2u-1)]^2  \right\}.
\]
Therefore the contribution to $\dot{y}_{jk}$ of the integrals in each subinterval $l=0,\ldots,2^{P+1}-1$ is either zero or $\pm 2^{j/2}$ times the precomputed value
\[
\int_{l/2^{P+1}}^{(l+1)/2^{P+1}} \!\!\!\!\!\! \lambda(u) du.
\]
Section~\ref{Sec:Results} presents the results obtained with the Wiener-Haar wavelet expansion in Eq.~\ref{eq:x_waveletexp} for hybrid dynamical systems.

\subsection{Modeling state resets}
\label{Sec:resets}


A significant challenge that hybrid dynamical systems present is
the possibility of state resets~\cite{SHS:book}. When a hybrid
system switches from one mode to another, the state of the
system can, in general, be reset discontinuously. For example,
in the case of the bouncing ball~\cite{SHS:book}, the velocity
of the ball changes discontinuously after every impact. When the
hybrid system transitions from mode $q$ to~$q'$ the state resets are typically represented as,
\begin{equation}
x^{+} = h(x^{-}),
\label{eq:reset}
\end{equation}
where $x^{-}$ and $x^{+}$ are the states of the system before and after the reset. Such discontinuities cannot be
easily accommodated within the hybrid polynomial chaos framework as described in the previous sections.

To circumvent this problem, one can construct a boundary layer in the vicinity of the guard condition. We also introduce a dummy vector ($z$) that tracks the state $x$ outside the boundary layer and is set to $x^{-}$ within the boundary layer. Note that we assume \emph{separability} of the states, i.e. the guard conditions (which determine the switching between modes of operation) can be written independently of the state reset conditions in Eqn.~\ref{eq:reset}. In other words, the states that determine the guard conditions do not participate in the state reset. Let the reset condition be in terms of vector $x$ in Eqn.~\ref{eq:reset}, and the guard conditions be in terms of vector $y$ (given by $\left\{y: g(y) = 0\right\}$). Note that, $\left[x \quad y\right]^{T}$ represents the entire state vector with the following governing equation,
\begin{eqnarray}
\dot x &=& f_{1}(x,y) \nonumber \\
\dot y &=& f_{2}(x,y).
\label{eq:sep_sys}
\end{eqnarray}

We now construct a boundary layer around the guard condition for vector $y$ as follows:

\begin{equation}
\begin{pmatrix}
\dot x\\
\dot y\\
\dot z
\end{pmatrix} = \begin{cases}
\begin{pmatrix}
f_{1}(x,y)\\
f_{2}(x,y)\\
(x - z)/\epsilon
\end{pmatrix} & \mbox{if}: |g(y)|\geq\epsilon \\
\\
\begin{pmatrix}
[h(z) - x]/\epsilon\\
\epsilon f_{2}(x,y)\\
0
\end{pmatrix} & \mbox{otherwise.}
\end{cases}
\label{eq:blayer}
\end{equation}

The above dynamical system is constructed such that $x^{-}$
evolves to $h(x^{-})$ in $\Delta t\approx\epsilon$, where $\epsilon$ is a small parameter.


By replacing each reset condition with an equation of the form
given by Eqn.~\ref{eq:blayer}, one obtains a new dynamical system without resets that approximates the original. On this new dynamical system one can use the expansion from Sec.~\ref{Sec:hpc} and evolve it using Eqn.~\ref{eq:intregionpc}. In other words, the framework
generalizing polynomial chaos to hybrid systems can be augmented using
Eqn.~\ref{eq:blayer} to include state resets.

To illustrate the procedure presented above we turn to the classic bouncing ball example~\cite{ZenoReg}: we consider the dynamics of a ball bouncing on a floor with coefficient of restitution $\gamma < 1$ under the action of gravity of uncertain magnitude ($\mu(g) = 9.8\,m/s^{2}$ and $\sigma(g) = 0.2\,m/s^2$). Thus, every time
the ball makes contact with the floor the velocity $v^{-}$ is reset to a new value given by $v^{+}=-\gamma v^{-}$. The guard condition for resetting the velocity is given by $y(t)=0$ (where $y(t)$ is the height of the ball above the floor at time $t$). The equations for the bouncing ball are given by,
\begin{eqnarray}
\dot y &=& v, \nonumber\\
\dot v &=& -g,
\label{eq:balleqs}
\end{eqnarray}
with the reset condition at $y=0$: $v^{+}=-\gamma v^{-}$. Thus, if one uses the boundary layer approximation in Eqn.~\ref{eq:blayer} we get,

\begin{equation}
\begin{pmatrix}
\dot y\\
\dot v\\
\dot z
\end{pmatrix} = \begin{cases}
\begin{pmatrix}
v\\
-g\\
(v - z)/\epsilon
\end{pmatrix} & \mbox{if}: |y|\geq\epsilon \\
\\
\begin{pmatrix}
\epsilon v\\
(\gamma z - v)/\epsilon \\
0
\end{pmatrix} & \mbox{otherwise.}
\end{cases}
\label{eq:blayer_ball}
\end{equation}
We now use the hybrid polynomial chaos expansion in Eqn.~\ref{eq:intregionpc} along with the Wiener-Haar wavelet basis functions. The Monte Carlo simulations on the bouncing ball are shown in Fig.~\ref{Fig:MonteBall}. The average or nominal trajectory is also shown. In Fig.~\ref{Fig:ByLayerBall} we compare the nominal trajectory (mean trajectory from Monte Carlo) with the mean predicted using the boundary layer expansion with $\epsilon=0.01$. As shown, the boundary layer accurately approximates the mean over multiple state resets events (in this case, each impact with the floor).

\begin{figure}
  \centering
   \subfigure[]{\includegraphics[scale=0.15]{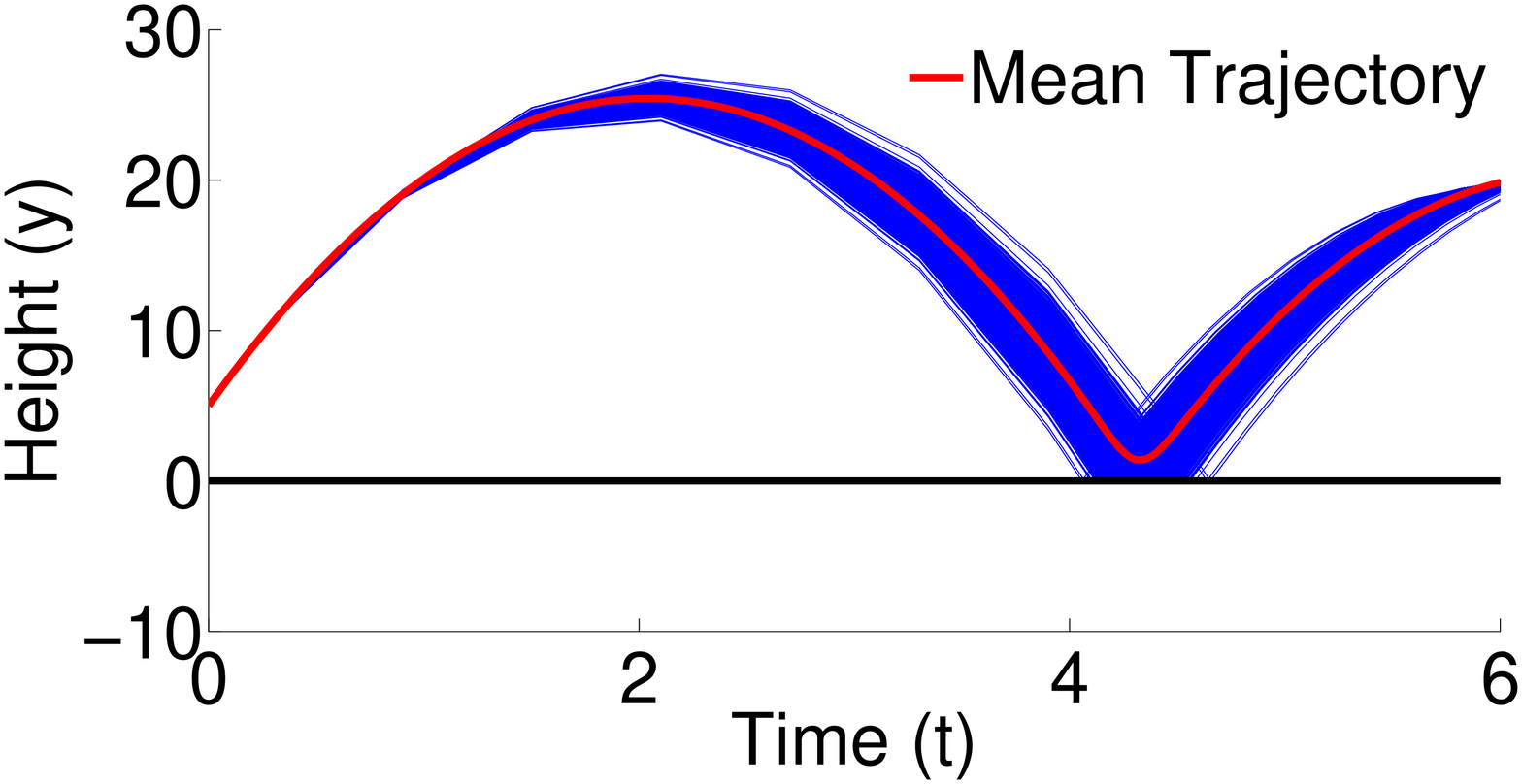}\label{Fig:MonteBall}}
   \subfigure[]{\includegraphics[scale=0.15]{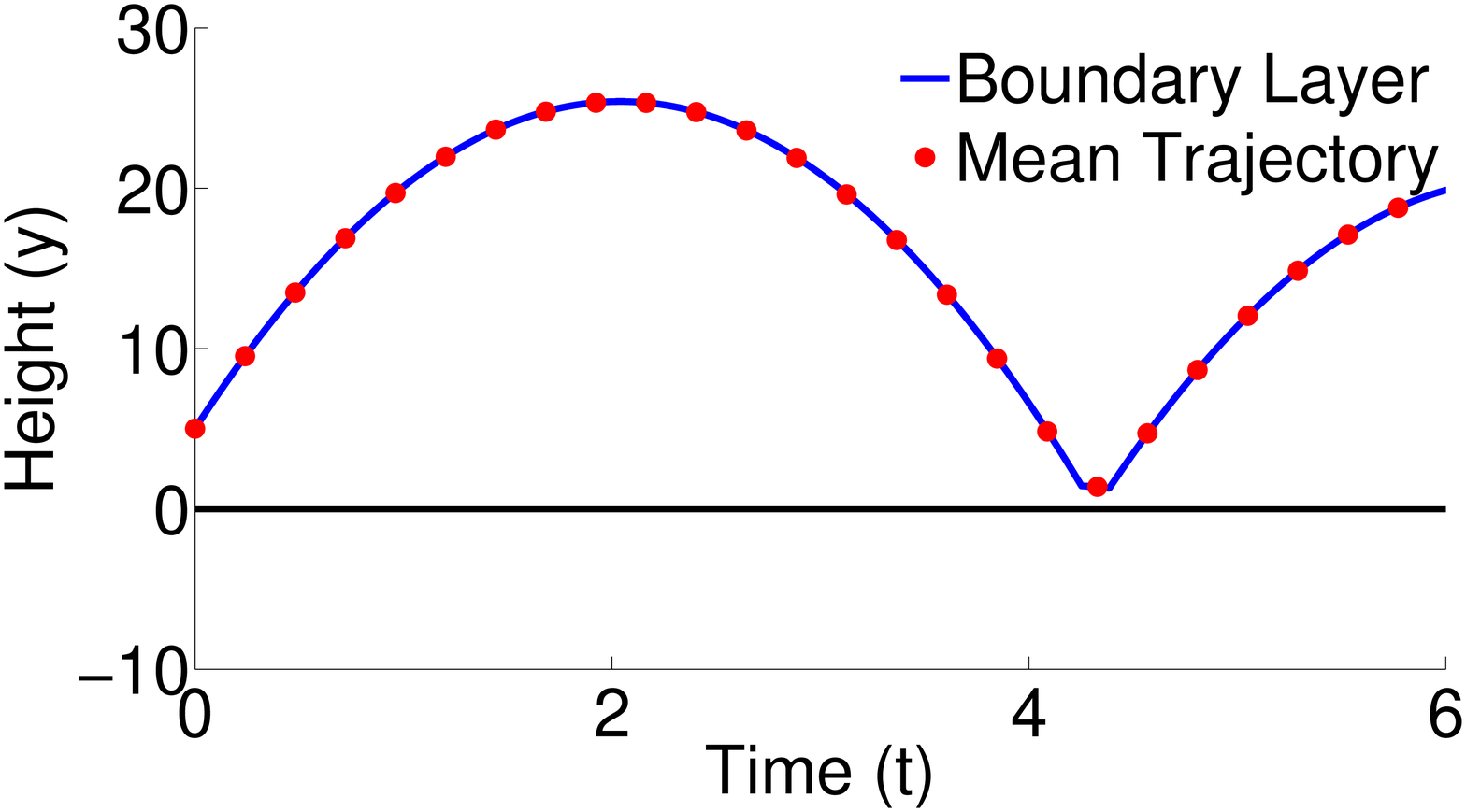}\label{Fig:ByLayerBall}}
  \caption{a) Monte Carlo simulation of the bouncing ball system with uncertain gravitational acceleration. b) Nominal bouncing ball trajectory compared with the mean trajectory obtained through the wavelet-based hybrid UQ approach with the boundary layer approximation ($\epsilon = 0.01$). \label{Fig:DistBall}}
 \end{figure}

\subsection{Results}\label{Sec:Results}

To demonstrate the hybrid polynomial chaos approach on hybrid dynamical systems we
consider the simple yet challenging example of a switching
oscillator given by Eqn.~\ref{eq:SHMswitch}.
\begin{align}
\ddot x + c\dot x + x + \lambda &= 0\,\,\,\text{if $x \geq 0$} \nonumber \\
\ddot x + c\dot x + x - \lambda &= 0\,\,\,\text{if $x < 0$}. \nonumber
\end{align}
The value of $c$ is deterministic and equal to~0.5. Here we
consider three cases with $\lambda$ normally distributed with:
$\mu(\lambda)=-10$ and $\sigma(\lambda)=2$ (case~1), $\mu(\lambda)=10$ and
$\sigma(\lambda)=2$ (case~2), and $\mu(\lambda)=0$ and $\sigma(\lambda)=1$ (case~3). In all cases we assume that the initial
conditions are deterministic and given by $\left[x(0),\dot
x(0)\right] = \left[10^{-2},1.0\right]$.

\subsubsection{Case 1: $\mu(\lambda)=-10$, $\sigma(\lambda)=2$}

Let us start with the case when $\mu(\lambda)=-10$ and
$\sigma(\lambda)=2$ in Eqn.~\ref{eq:SHMswitch}. A representative
trajectory for the dynamics of the system is shown in
Fig.~\ref{Fig:Trajcase1}. The corresponding histogram for
$x(20.0;\lambda)$ is shown in Fig.~\ref{Fig:Histcase1}. Most
importantly, one desires to compute the mean and variance of
$x(t;\lambda)$ as a function of time. In the system given by
Eqn.~\ref{eq:SHMswitch}, we expand $x(t;\lambda)$ using
Eqn.~\ref{eq:expx} and perform a Galerkin projection as shown in
Eqn.~\ref{eq:intregionpc}. One then gets a system of equations
for the coefficients of expansion in
Eqn.~\ref{eq:expx}. These coefficients, once computed, can be
used to calculate the moments of the distribution of
$x(t;\lambda)$.

We compare the results obtained from hybrid polynomial chaos with those
obtained using Monte Carlo and Quasi-Monte Carlo based methods.
In particular, we use a Weyl sequence~\cite{Niederreiter1992}
along with inverse transform sampling~\cite{InvTran:book} to
generate the Quasi-Monte Carlo samples. We find that the results
(in the first two moments) from $5000$ samples of Monte Carlo,
$3000$ samples of Quasi-Monte Carlo and the Wiener-Haar hybrid PC
expansion with $P=3$ are visually indistinguishable (see
Figs.~\ref{Fig:Meanxcase1} and~\ref{Fig:Varxcase1}). Treating
$5000$ Monte Carlo samples as baseline, we find that the hybrid PC expansion has a maximum error of $5\times
10^{-2}$ in the prediction of $\mu(x(t;\lambda))$.

%
\begin{figure}
  \centering
  \subfigure[]{\includegraphics[scale=0.4]{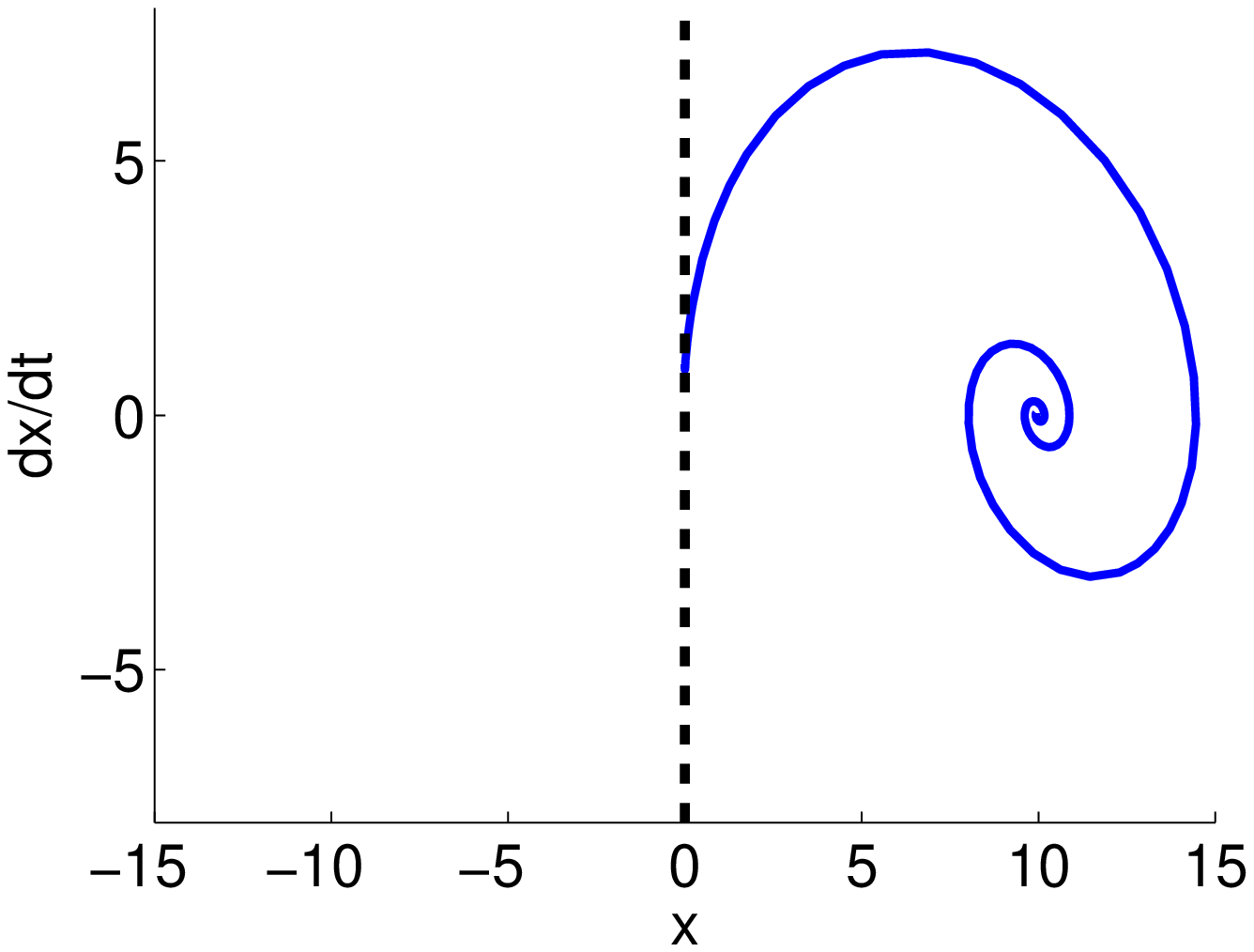}\label{Fig:Trajcase1}}
  \subfigure[]{\includegraphics[scale=0.4]{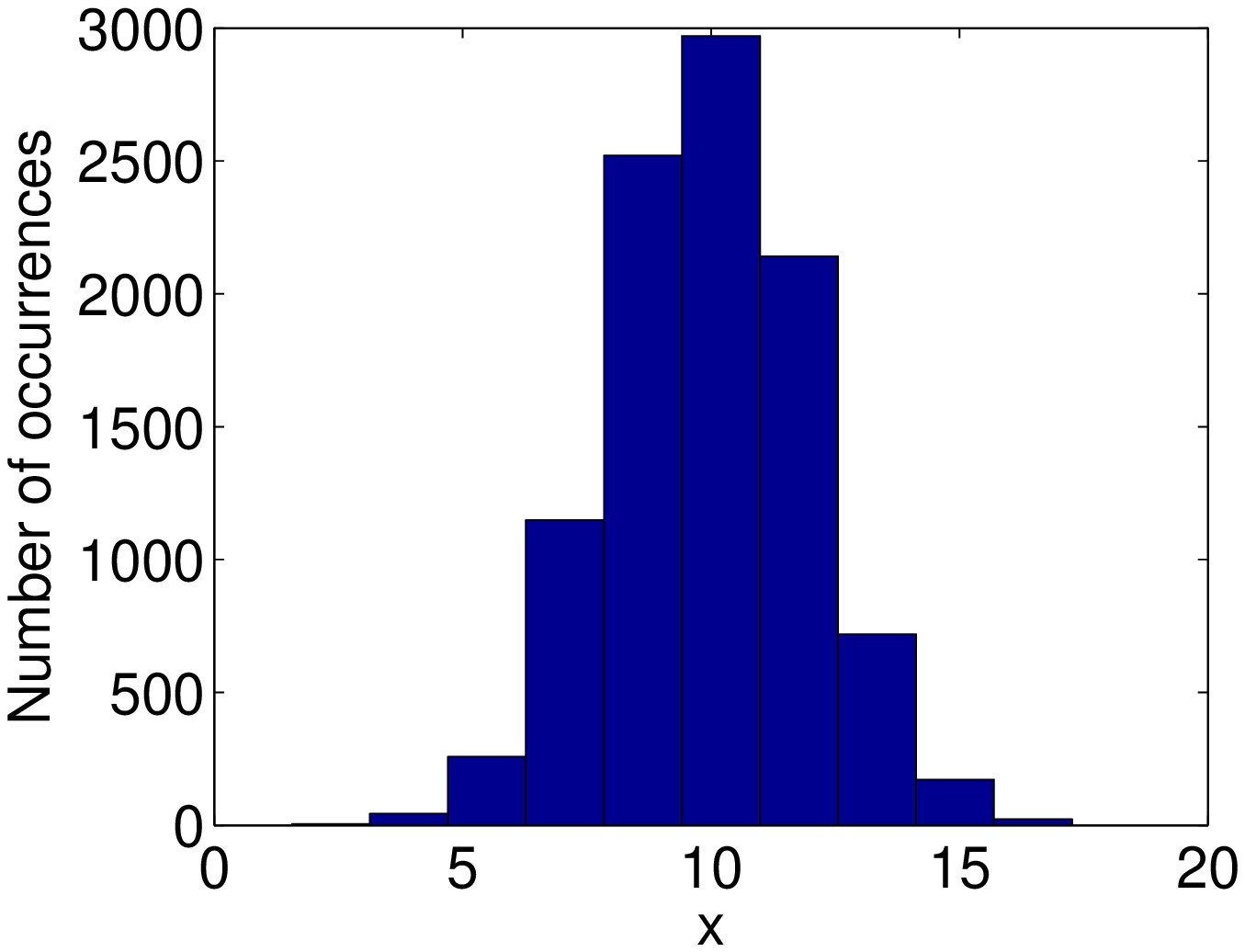}\label{Fig:Histcase1}}
  \caption{a) A representative trajectory for case $1$. b) Histogram of $x(20;\lambda)$ for case $1$.}
 \end{figure}

%
\begin{figure}
  \centering
  \subfigure[]{\includegraphics[scale=0.3]{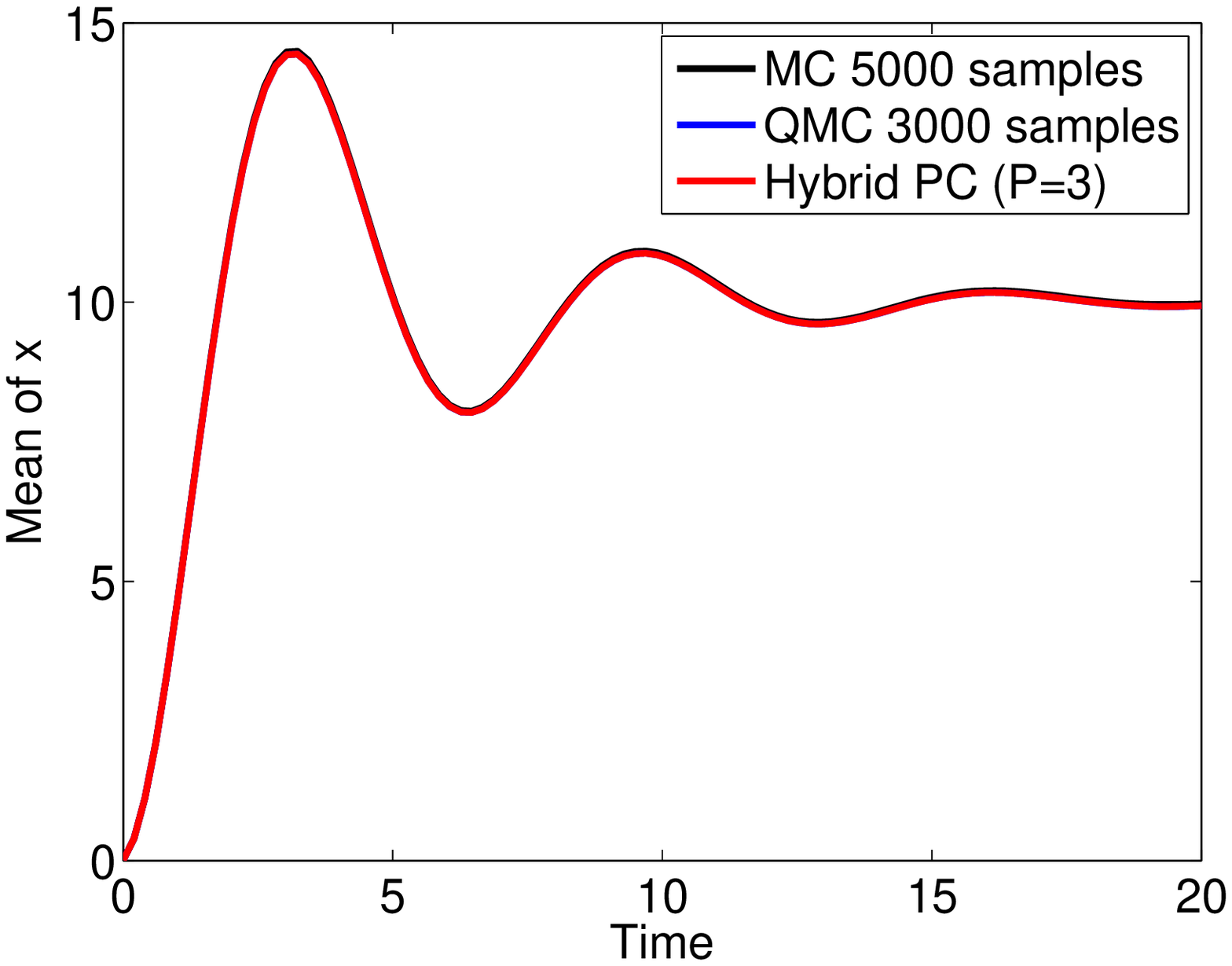}\label{Fig:Meanxcase1}}
  \subfigure[]{\includegraphics[scale=0.3]{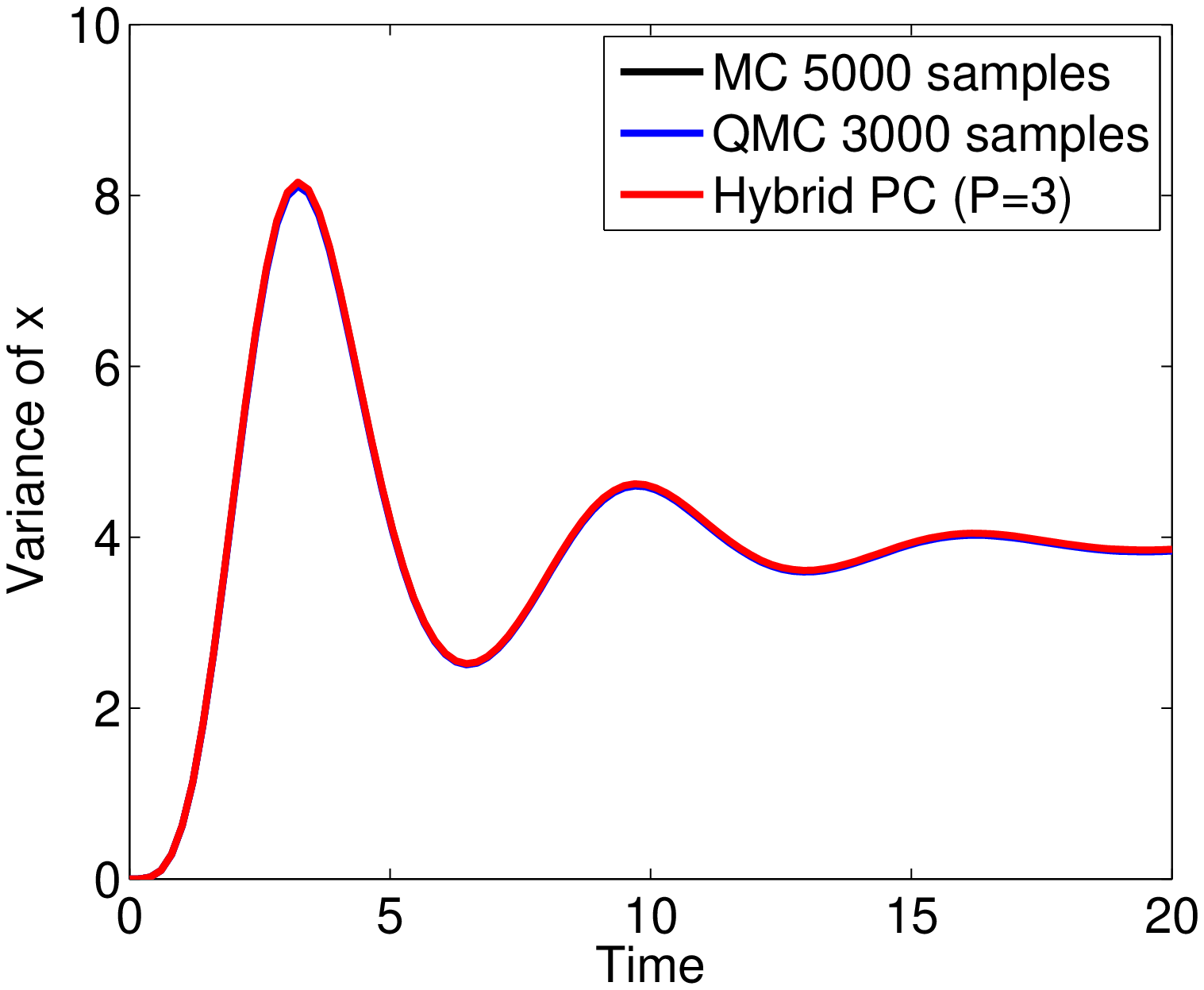}\label{Fig:Varxcase1}}
  \caption{ Comparison of a) predicted mean and b) predicted variance of $x(t;\lambda)$ for various UQ methods for case $1$.}
 \end{figure}

\subsubsection{Case 2: $\mu(\lambda)=10$, $\sigma(\lambda) = 2$}

The case of $\mu(\lambda)=10$ and $\sigma(\lambda)=2$ is
significantly more challenging. A representative trajectory of
the system (for $\lambda = 10$) is shown in
Fig.~\ref{Fig:Trajcase2}. When $\lambda > 0 $ the system
switches back and forth between modes. The reason for this is that when the system is in the right half-plane, the equilibrium
of the system is in the left half-plane and vice versa. A
histogram for $x(3.0;\lambda)$ is depicted in
Fig.~\ref{Fig:Histcase2}.

We again compare hybrid Wiener-Haar polynomial chaos to Monte Carlo sampling in
Fig.~\ref{Fig:case2}. We find that hybrid Wiener-Haar polynomial chaos
($p=5$) accurately computes the mean and the variance of the distribution of $x$. The maximum absolute error of hybrid polynomial chaos in mean is $\mu(x(t;\lambda) = 1.8\times10^{-3}$ and variance is $7\times10^{-5}$. Note that expansions
in terms of standard basis functions such as Hermite and Legendre polynomials
are unable to compute the moments of $x(t;\lambda)$ beyond a
threshold time that depends weakly on the order of
expansion~$p$ (see Fig.~\ref{Fig:failcase2}). The solution in this case is particularly challenging because it becomes more oscillatory in
terms of $\lambda$ at $t$ increases (Fig.~\ref{Fig:Oscfail}). The Wiener-Haar basis functions are naturally oscillatory and hence more accurate than Hermite polynomials in capturing the solution $x(t;\lambda).$ Note that, for large time simulations the Wiener-Haar
expansions will also fail since the solution will eventually become too oscillatory for the order of expansion. This problem is well known in the polynomial chaos literature~\cite{Longterm}.

\begin{figure}
  \centering
  \subfigure[]{\includegraphics[scale=0.4]{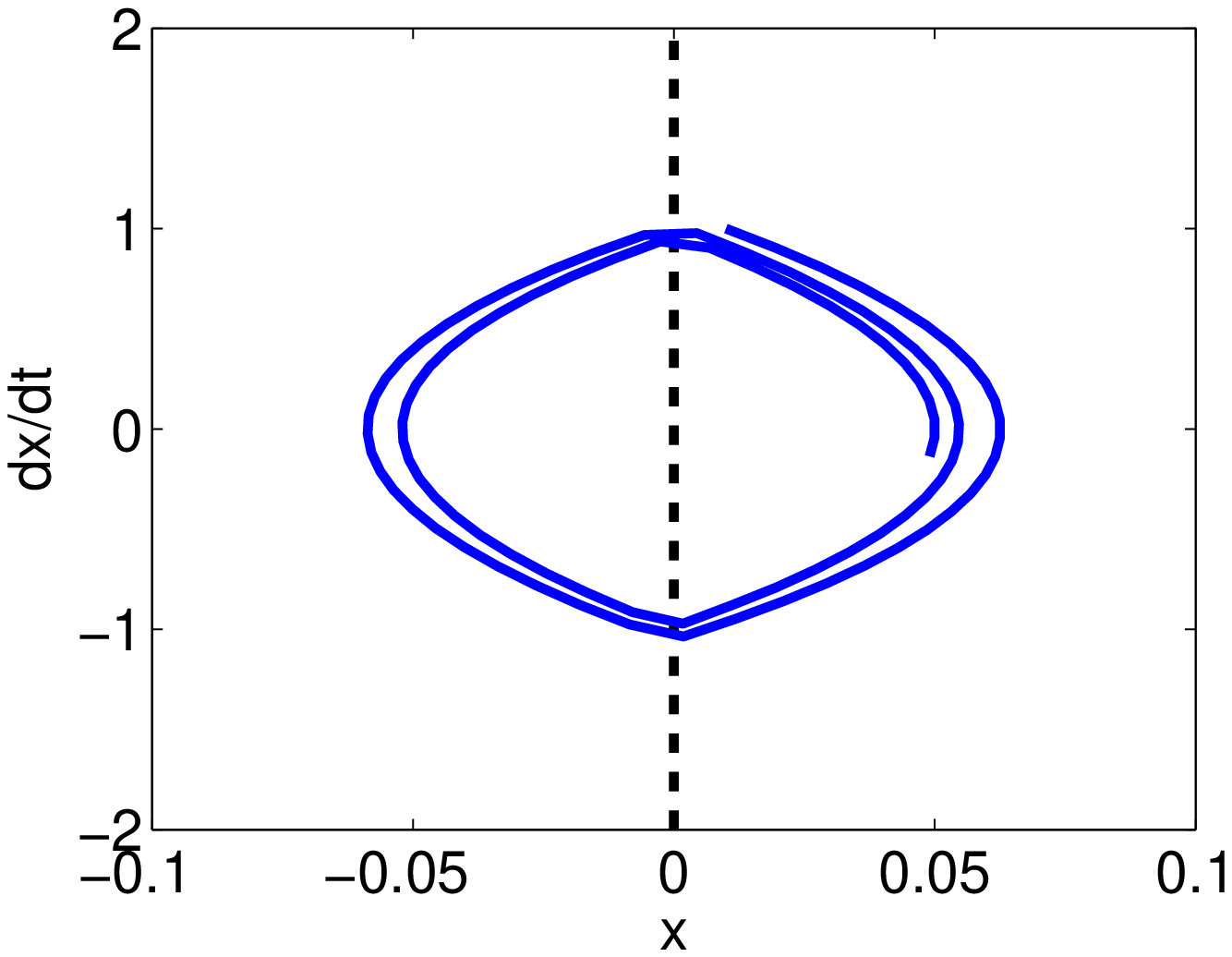}\label{Fig:Trajcase2}}
  \subfigure[]{\includegraphics[scale=0.4]{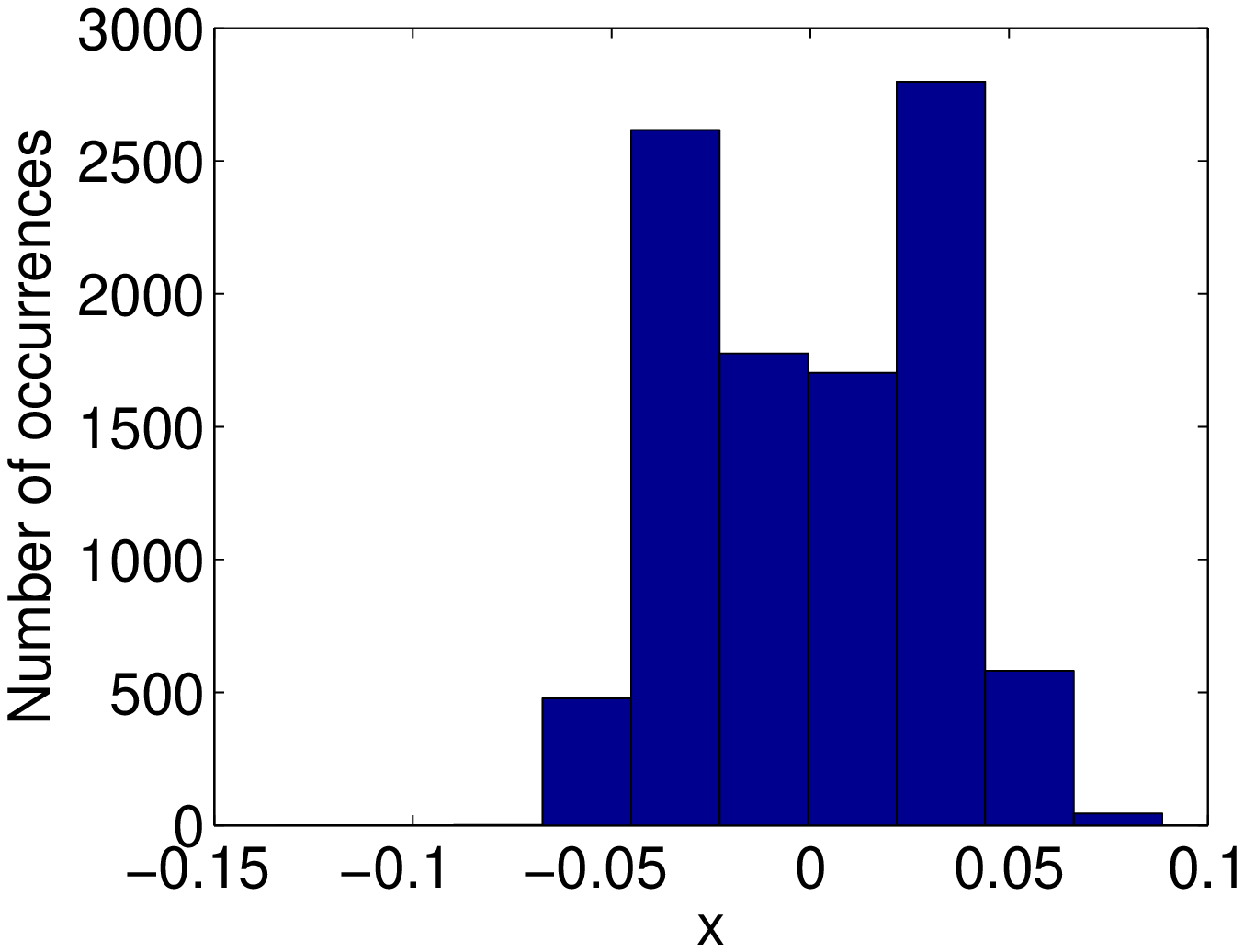}\label{Fig:Histcase2}}
  \caption{a) A representative trajectory for case $2$. b) Histogram of $x(3.0;\lambda)$ for case~2.\label{Fig:Dycase2}}
 \end{figure}

\begin{figure}
  \subfigure[]{\includegraphics[scale=0.4]{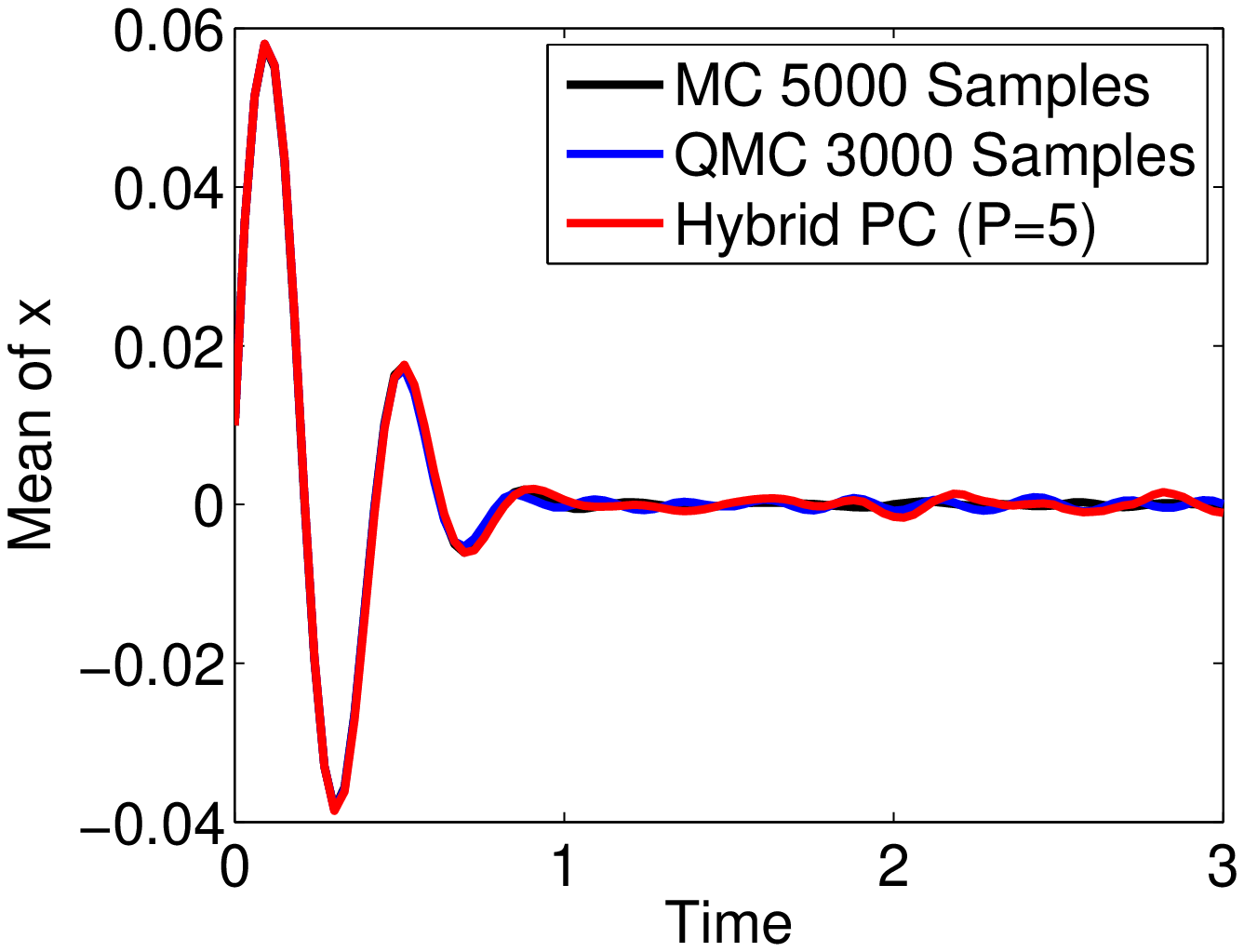}}
  \subfigure[]{\includegraphics[scale=0.4]{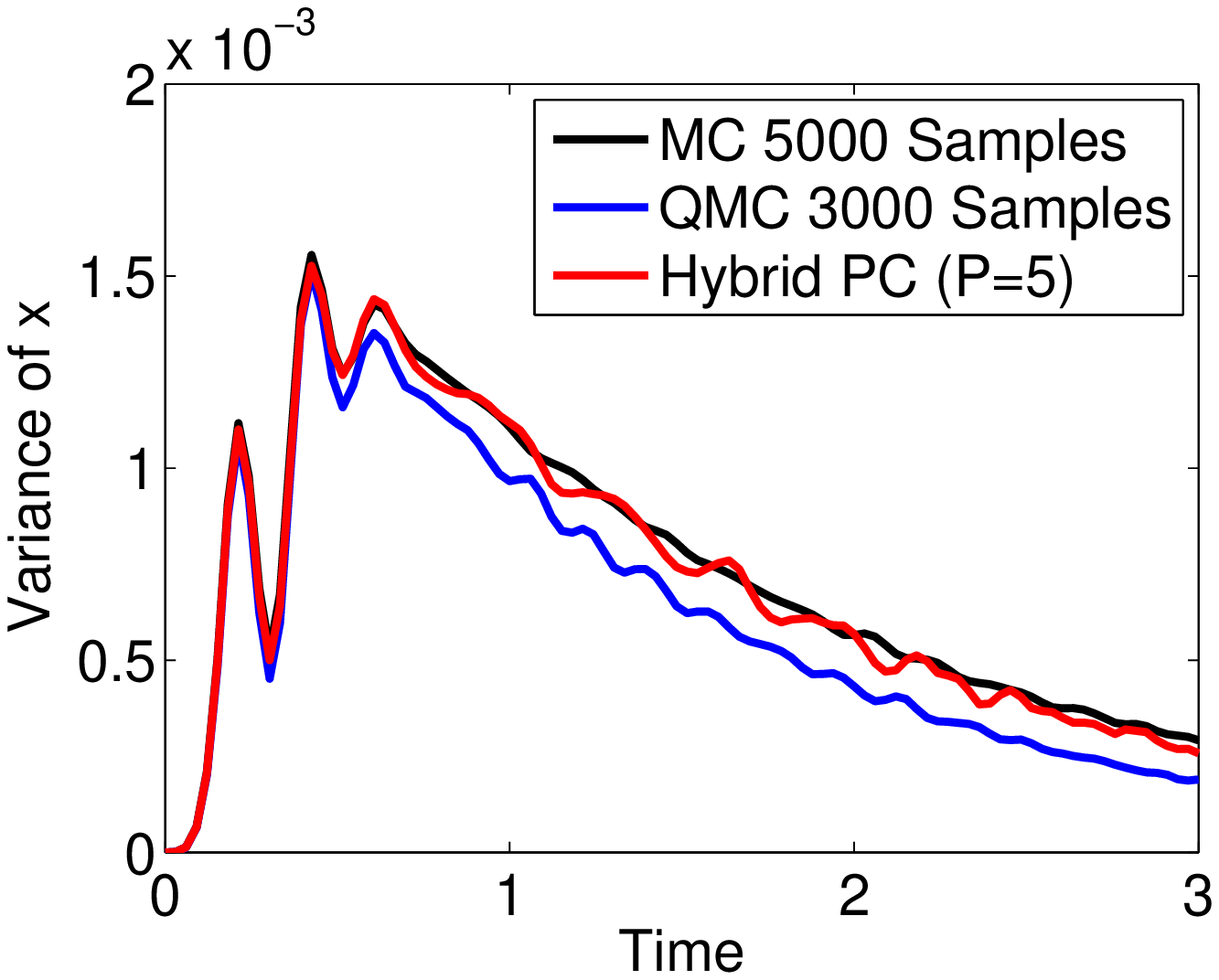}\label{Fig:case2var}}
  \caption{Comparison of a) predicted mean and b) predicted variance of $x(t;\lambda)$ by various UQ methods for case $2$. \label{Fig:case2}}
\end{figure}


\begin{figure}
  \centering
  \subfigure[]{\includegraphics[scale=0.3]{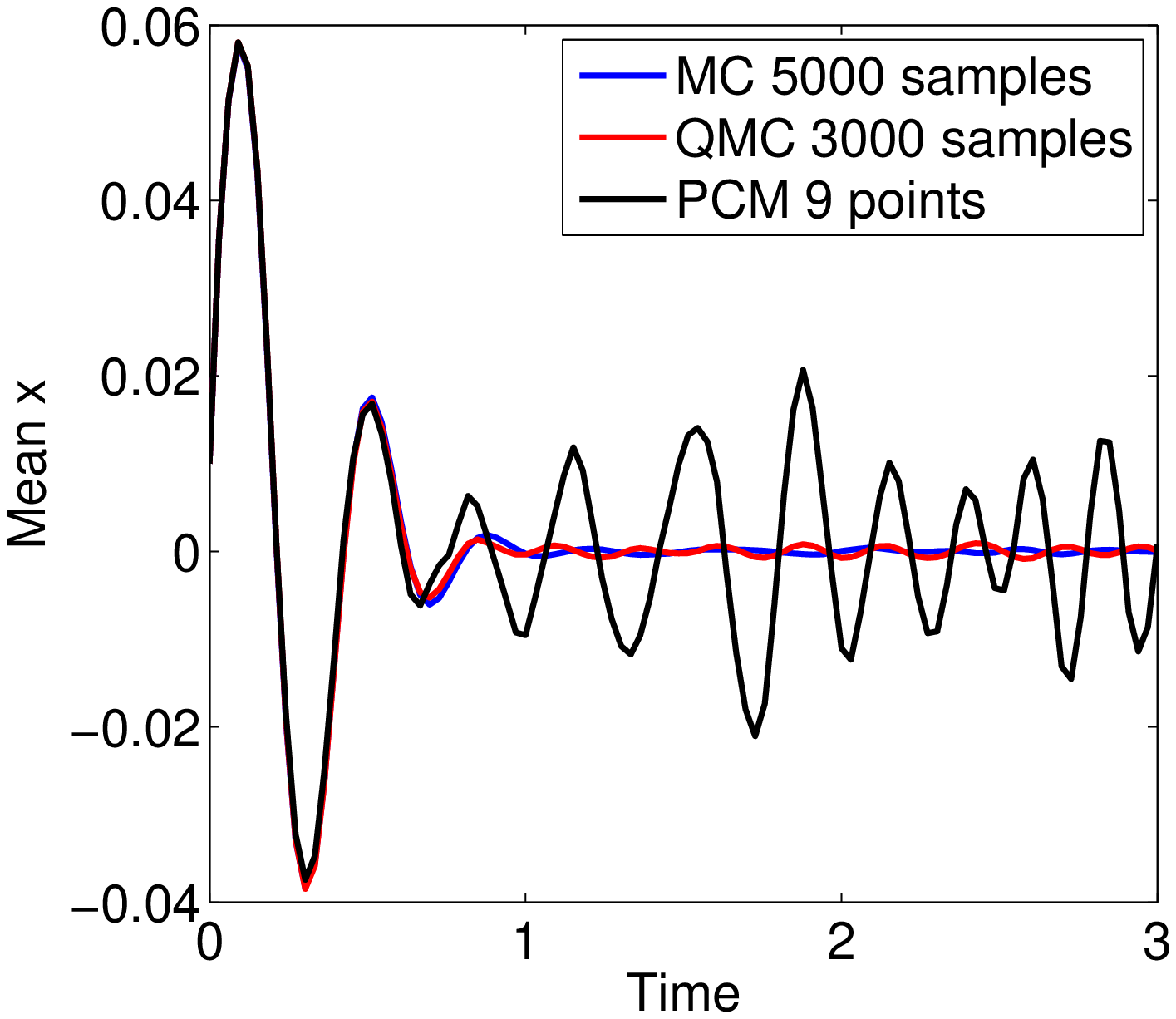}\label{Fig:failcase2mean}}
  \subfigure[]{\includegraphics[scale=0.3]{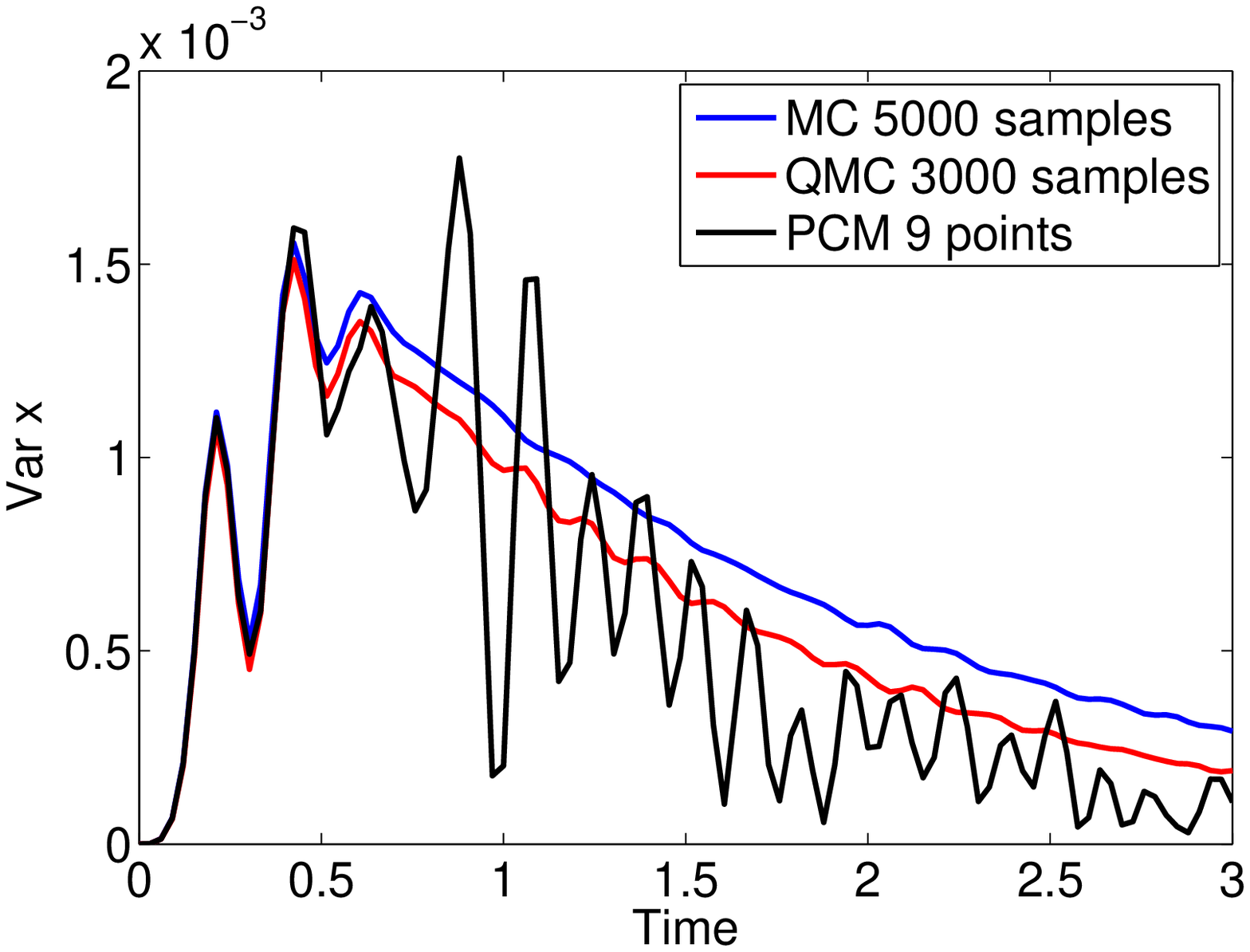}\label{Fig:failcase2var}}
  \caption{a) Mean and b) Variance predicted by standard (Hermite) polynomial chaos basis functions for case $2$.\label{Fig:failcase2}}
 \end{figure}
%

%

\begin{figure}
  \centering
  \includegraphics[width=0.5\hsize]{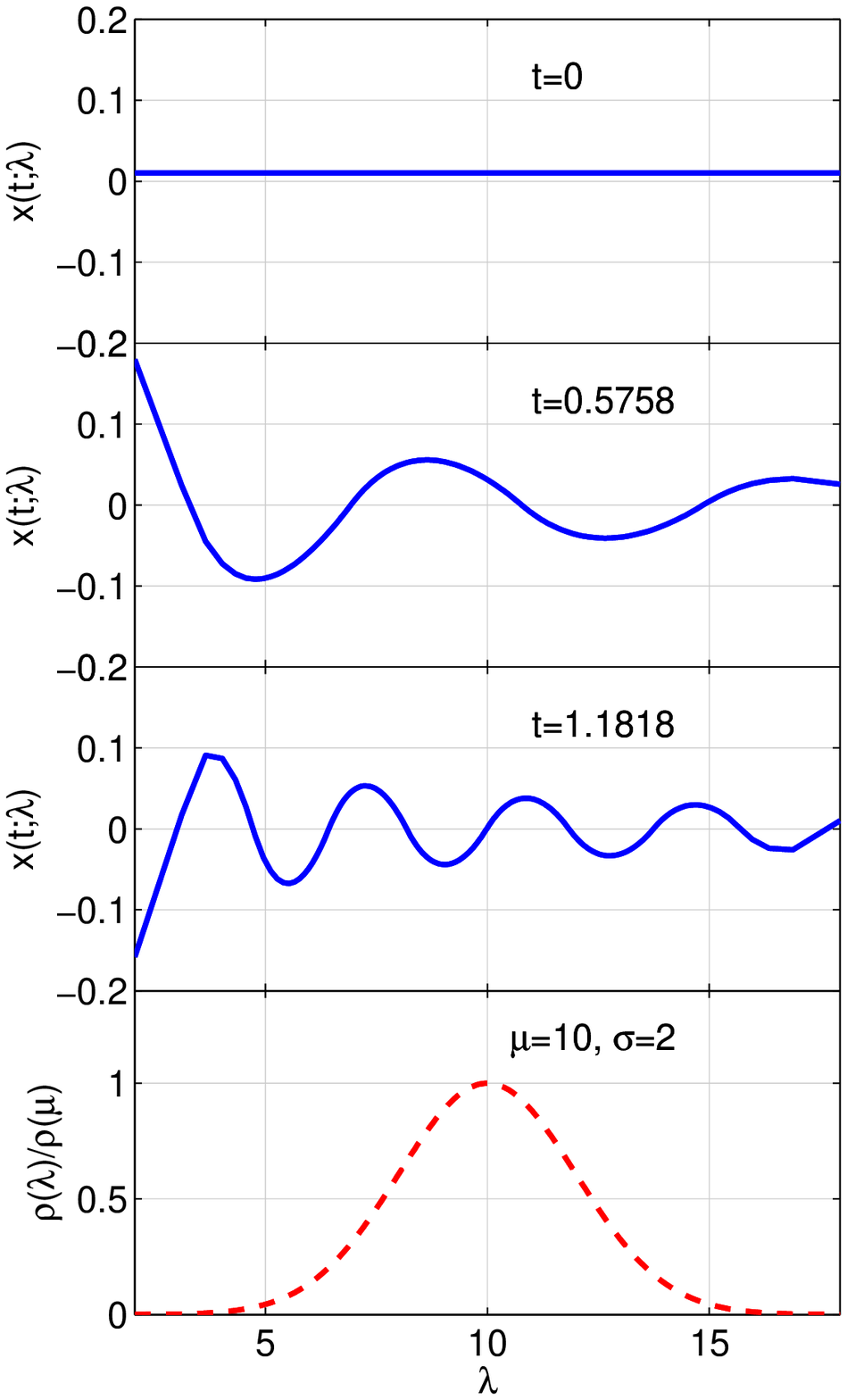}
    \caption{Case~2: $x(t;\lambda)$ becomes more oscillatory in $\lambda$ as $t$ increases.
    The bottom plot shows the distribution for~$\lambda$. \label{Fig:Oscfail}}
\end{figure}


%

\subsubsection{Case 3: $\mu(\lambda)=0$, $\sigma(\lambda) = 1.0$}

We also consider the case of $\mu(\lambda)=0$ with $\sigma(\lambda)=1.0$. This case is particularly challenging because there is a concentration of probability of $x(t;\lambda)$ as shown in Fig.~\ref{Fig:Histcase3}. The reason for this is as follows: when $\mu(\lambda)=0$, the nominal trajectory converges to $0$ and so do all trajectories with $\lambda>0$. Indeed, for trajectories with $\lambda>0$ the equilibrium lies in the opposite half-plane with respect to the current state. This gives rise to decaying switching trajectories, as case 2 in Fig.~\ref{Fig:Trajcase2}. Note that for $\lambda<0$, the trajectories are similar to the ones in case 1 (Fig.~\ref{Fig:Trajcase1}).

The Wiener-Haar basis functions along with the hybrid polynomial chaos approach accurately capture the moments of the distribution for $x(t;\lambda)$ (see Fig.~\ref{Fig:case3}). In fact, an expansion to just $P=3$ captures the first two moments. The step-function nature of the Wiener-Haar basis allows it to perform well in this scenario. Standard basis functions like Hermite polynomials are completely incapable of accurately capturing the moments of the distribution for $x(t;\lambda)$ shown in Fig.~\ref{Fig:Histcase2}.

\begin{figure}
  \centering
  \subfigure[]{\includegraphics[scale=0.4]{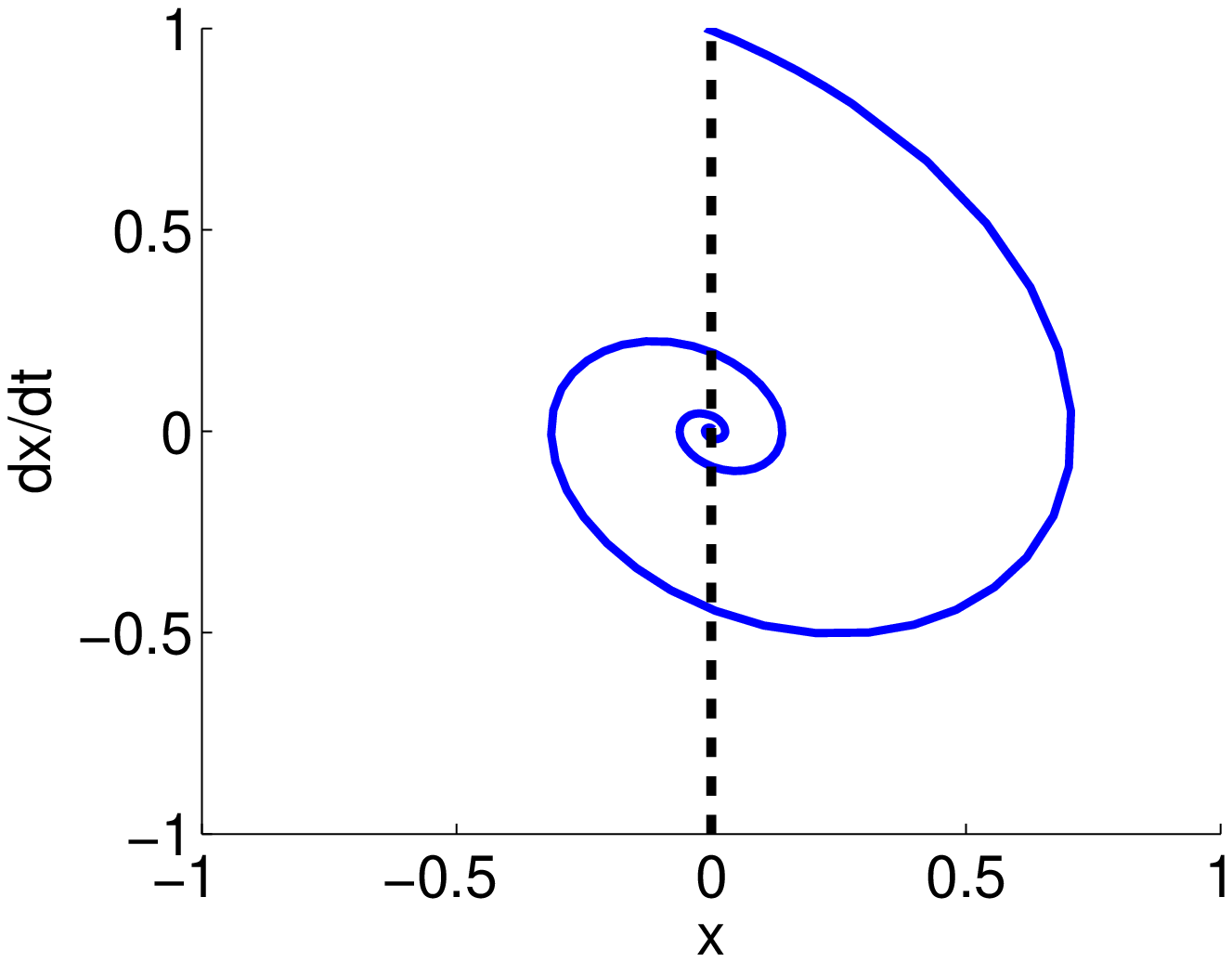}\label{Fig:Trajcase3}}
  \subfigure[]{\includegraphics[scale=0.4]{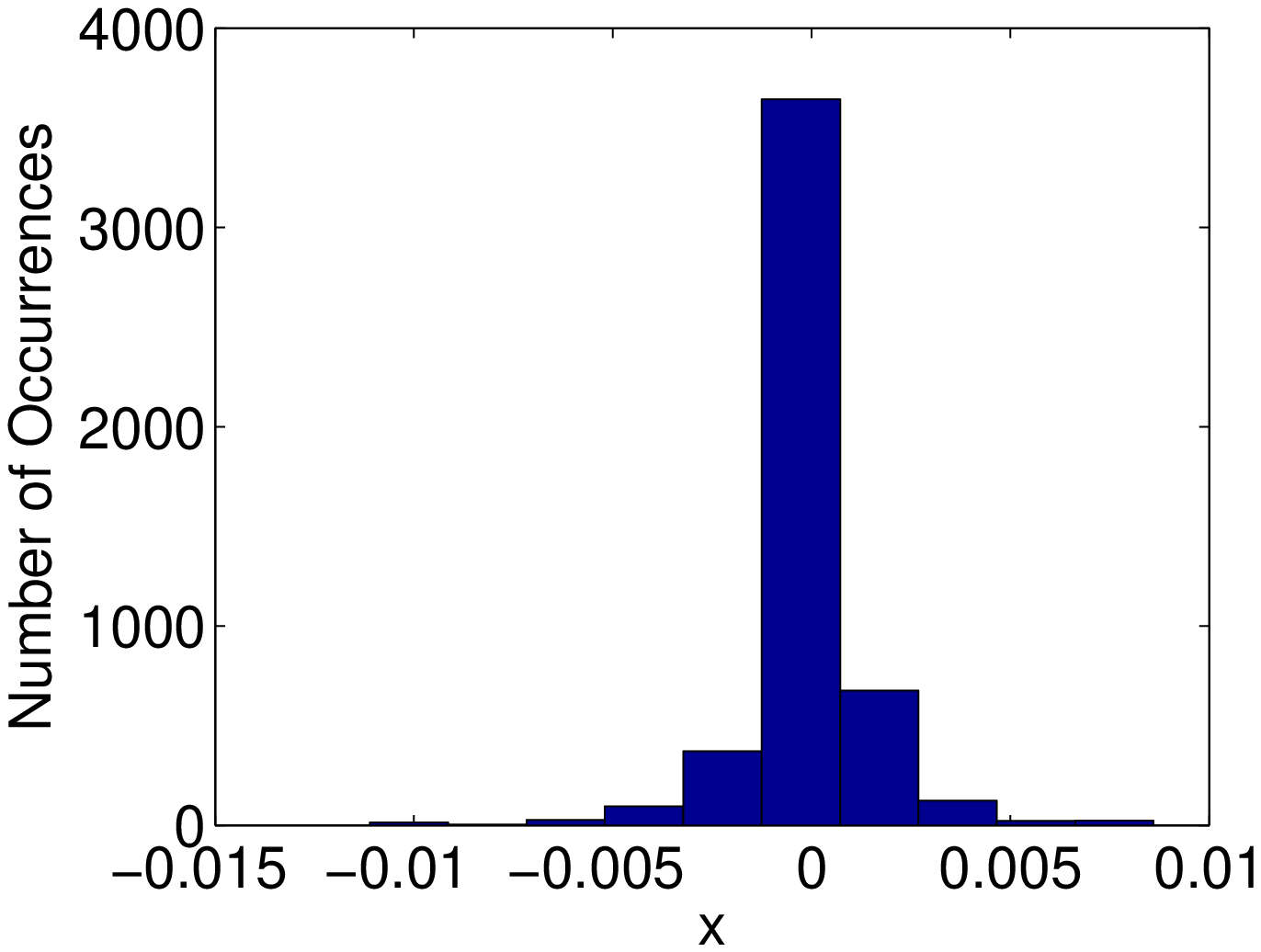}\label{Fig:Histcase3}}
  \caption{a) Nominal trajectory for case~3. b) Histogram of $x(20;\lambda)$ for case~3.}
 \end{figure}

\begin{figure}
  \subfigure[]{\includegraphics[scale=0.4]{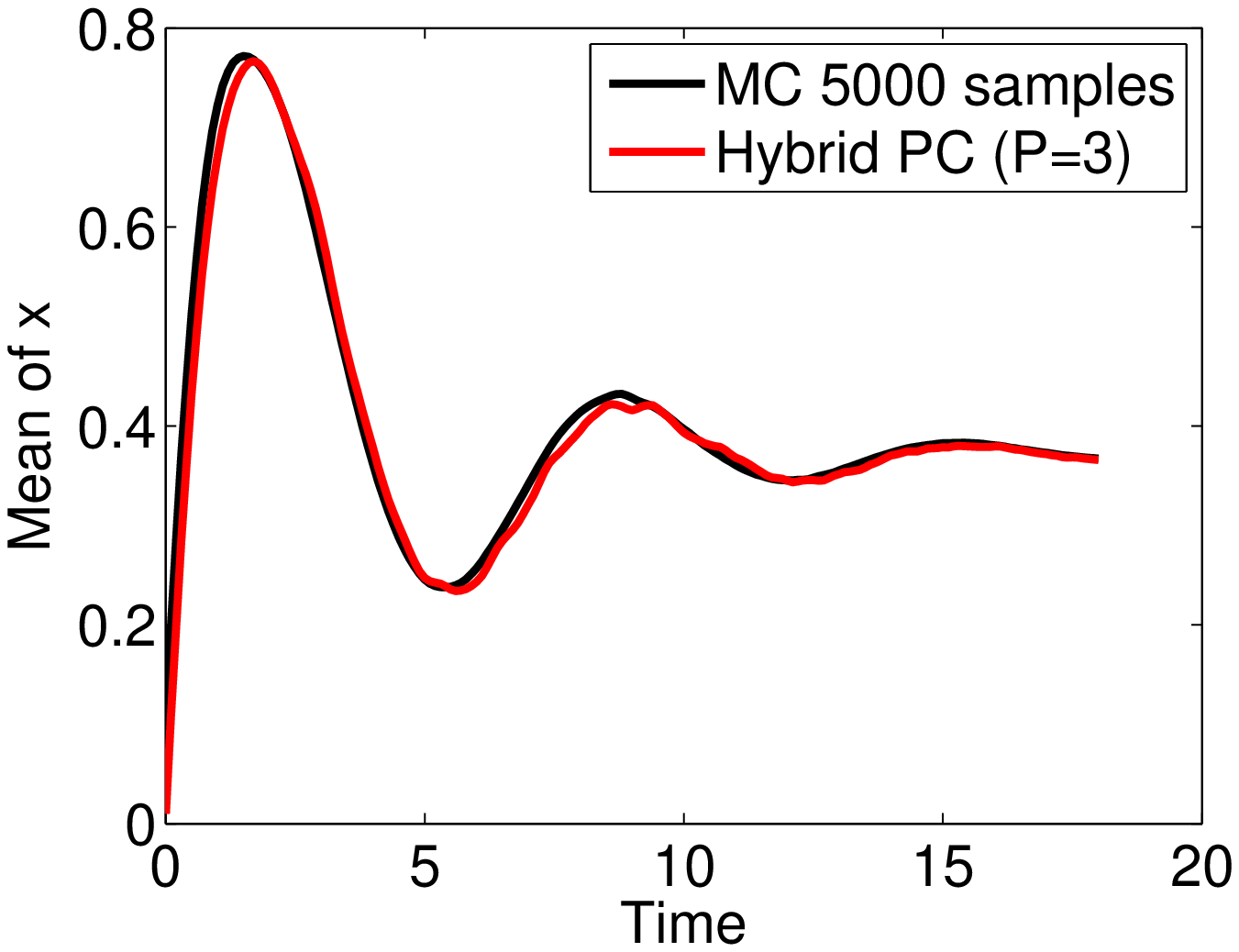}}
  \subfigure[]{\includegraphics[scale=0.4]{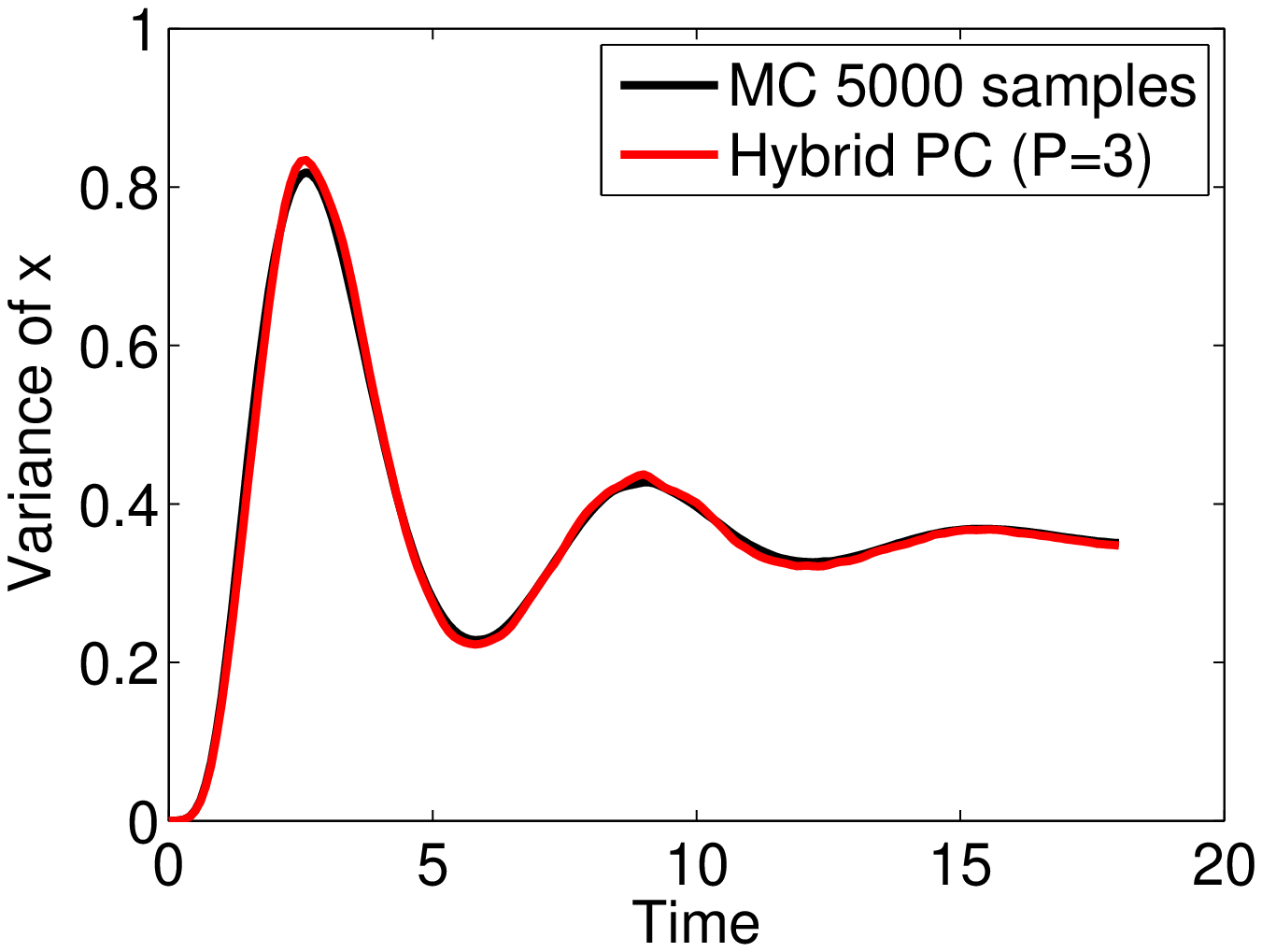}}
  \caption{Comparison of a) predicted mean and b) predicted variance of $x(t;\lambda)$ for various UQ methods for case $3$.\label{Fig:case3}}
\end{figure}

\section{Transport theory approach for uncertainty quantification in hybrid systems}
\label{sec:transporttheory}

In this section we present a qualitatively different approach to UQ in hybrid systems based on transport equations. We write an advection equation for the probability density of the state and expand this equation in an appropriate basis, as is done in polynomial chaos. The resulting equation is equivalent to the Fokker-Planck equation~\cite{Cit:FPE:book} in the absence of a diffusion term. Though significant effort has been put into computing solutions for the Fokker-Planck equation in various applications~\cite{Cit:FPE:book,Cit:FPE:book2}, our setting is particularly challenging due to the switching dynamics of hybrid systems. We note that advection equations for probability distribution functions have been used to propagate uncertainty through heterogeneous porous media with uncertain properties~\cite{Tartakovsky2011} and for hyperbolic conservation laws with noise~\cite{Luo2006}. Recently, similar methods have been extended to cumulative distribution functions in hyperbolic conservation laws~\cite{Wang2012}.

The polynomial chaos expansion in this setting yields a system of hyperbolic partial differential equations for the coefficients of the expansion, which are then solved by integrating along characteristics. The hybrid nature of the original system is reflected in that the characteristics exhibit switching. Even though we only consider systems without resets, we can use the results of Sec.~\ref{Sec:resets} to treat systems with resets.

As in Sec.~\ref{Sec:hyb_poly}, let us consider a hybrid system without resets and uncertain parameters~$\lambda$ with guard conditions independent
of~$\lambda$:\footnote{Note that this embodies the constraint of having no overlap in the domains for different modes.}
\[
\dot{x} = f_i(x,\lambda) \quad \text{when $G_i(x)$ is true.}
\]
The system has uncertain initial conditions described by the
probability density~$\rho_{x0} (x)$ and the uncertain parameters
$\lambda$ follow~$\rho_\lambda (\lambda)$.

We describe the system by the time evolution of the distribution
function $\rho(x,\lambda;t)$, which has initial condition
\[
\rho(x,\lambda;0) = \rho_{x0}(x) \rho_\lambda(\lambda) \quad
\text{(initial uncertainties are independent)}
\]
and normalization
\[
\int \rho(x,\lambda;t) dx d\lambda = 1.
\]
Note that, for all time, we have
\begin{equation}
\rho_\lambda(\lambda) = \int \rho(x,\lambda;t) dx.
\label{rho_lambda_is_constant}
\end{equation}

Our goal is to compute the evolution of the density in~$x$ (the
marginal distribution):
\[
\rho_x(x,t) = \int \rho(x,\lambda;t) d\lambda.
\]
However, without introducing assumptions on $\rho_\lambda$, the equation for $\rho_x$ is not closed. We therefore focus on computing the evolution of $\rho$ directly through an expansion. From this evolution, $\rho_x$ can then be calculated at every instant.

\subsection{Equation for $\rho$}

Let us define the sets $S_i = \left\{ x \; : \; G_i(x) \;
\text{is true} \right\}$ and the indicator functions
\[
\mathbf{1}_i(x) = \begin{cases} 1 & \text{if $x \in S_i$} \\ 0 &
\text{if not}. \end{cases}
\]
With this notation, and because $\lambda$ is constant along a
trajectory, we have
\begin{equation}
\frac{\partial \rho}{\partial t} + \nabla \cdot (\rho f) = 0
\qquad \forall \lambda \label{liouville}
\end{equation}
where $f(x,\lambda) = \sum_i \mathbf{1}_i(x) f_i(x,\lambda)$ and the gradient operator acts only on~$x$ and not on~$\lambda$.

\subsection{Boundary conditions at the interfaces}

The discontinuity in the equation implies that mass may
accumulate at the boundaries between zones where different guard
conditions are valid. Integrating on a cylinder that
crosses one such boundary we obtain the matching condition
\begin{equation}
\frac{\partial \sigma}{\partial t} + \nabla_s \cdot (\sigma f) =
\rho_i f_i \cdot \hat{n}_{ik} - \rho_k f_k \cdot \hat{n}_{ik},
\label{boundary_condition}
\end{equation}
where $\sigma$ is a surface probability density between regions
$S_i$ and~$S_k$, $\hat{n}_{ik}$ is the surface normal from $S_i$
to~$S_k$, $\nabla_s$ is the divergence in the space tangent to
the surface, and $f$ is the flow \emph{at the
surface}.\footnote{Whether $f$ is $f_i$ or $f_k$ on the surface
will depend on how the guard conditions are expressed.} This may
lead to a cascade, with probability condensing into
progressively lower dimensional structures: where two
hypersurfaces meet (the boundary between three guard conditions)
the same scenario repeats, until we have mass accumulating at
points. To solve Eqn.~\ref{liouville} the initial condition must
include initial values for~$\sigma$ and the
probability density on any lower dimensional structure where
mass may accumulate.

The discontinuity of $f$ does not \emph{necessarily} imply
accumulation. In fact, at an interface we have several options:
\begin{enumerate}
    \item Both $f_i \cdot \hat{n}_{ik}$ and $f_k \cdot \hat{n}_{ik}$ are nonzero and have the same sign. In this case, there is no accumulation. If we assume that $\rho$ has a singularity at the interface, the flow will move the singularity away from it.

    \item $f_i \cdot \hat{n}_{ik} > 0$ and $f_k \cdot \hat{n}_{ik} \leq 0$: accumulation occurs.

    \item $f_i \cdot \hat{n}_{ik} \geq 0$ and $f_k \cdot \hat{n}_{ik} < 0$: accumulation occurs.

    \item $f_i \cdot \hat{n}_{ik} = f_k \cdot \hat{n}_{ik} = 0$: no accumulation.

    \item $f_i \cdot \hat{n}_{ik} \leq 0$ and $f_k \cdot \hat{n}_{ik} \geq 0$: no accumulation.
\end{enumerate}
If we are in the case without accumulation and without initial concentration of density in lower dimensional structures, then $\sigma=0$ and
Eqn.~\ref{boundary_condition} becomes
\begin{equation}
\rho_i f_i \cdot \hat{n}_{ik} = \rho_k f_k \cdot \hat{n}_{ik}.
\label{flow_matching}
\end{equation}

\begin{thm}
\label{thm:noaccumulation}
Any second order ODE of the form
\[
\ddot{X} = F_i(X,\dot{X},\lambda) \quad \text{when $G_i(X)$ is
true}
\]
satisfies the conditions for no accumulation at the interface
where the ODE is discontinuous.
\end{thm}

\begin{proof}

Without loss of generality we can focus on just two regions
$S_i$ and~$S_k$ and rewrite the problem as the first-order ODE
\[
\dot{x} \equiv \left(\begin{matrix} \dot{X} \\ \dot{Y}
\end{matrix} \right) = f(x) = \left(\begin{matrix} Y \\
F_i(X,Y,\lambda) \end{matrix} \right) \quad \text{when $G_i(X)$
is true.}
\]
To find the normal $\hat{n}_{ik}$ we consider a ${\cal C}^1$
function $b(x) = b(X,Y) = B(X)$ that is positive in $S_k$ and
negative $S_i$ so that the interface is given by the locus of
$b(x)=0$. The gradient of this function is proportional to the
normal:
\[
\hat{n}_{ik} \propto \nabla b = \left(\begin{matrix} \nabla_X B
\\ 0 \end{matrix} \right)
\]
and therefore
\[
\hat{n}_{ik} \cdot f = \frac{Y \cdot \nabla_X B}{||\nabla_X
B||},
\]
which is continuous at the interface.

\end{proof}

\subsection{Expansion of the equation for $\rho$}

At every point $x$ we expand the distribution in $\lambda$:
\begin{equation}
\rho(x,\lambda;t) = \sum_k  a_k(x,t) w(\lambda) \psi_k(\lambda),
\label{expansion}
\end{equation}
where $\{\psi_k\}$ forms an orthogonal basis with respect to
$w$:
\[
\int \psi_i(\lambda) \psi_k(\lambda) w(\lambda) d\lambda = w_k
\delta_{ik}.
\]
We keep $w_k$ to allow for a non-normalized weight function
$w(\lambda)$. We replace the expansion in Eqn.~\ref{expansion}
into Eqn.~\ref{liouville} and project onto $\psi_i$ to obtain a set of partial differential equations for the coefficients $a_i(x,t)$:
\begin{equation}
\frac{\partial a_i(x,t)}{\partial t} + \frac{1}{w_i}\nabla \cdot
\sum_k a_k(x,t) \int w(\lambda) \psi_i(\lambda) \psi_k(\lambda)
f(x,\lambda) d\lambda = 0. \label{liouville_expanded}
\end{equation}
Since this equation is local in~$x$ there is no question as to
which $f(x,\lambda)$ must be used at any given point.

\subsection{Example: switching oscillator}

Here we revisit the switching oscillator system
\[
\ddot{x} = \begin{cases} -x - \gamma \dot{x} - \lambda & \text{if $x \geq 0$} \\
-x - \gamma \dot{x} + \lambda & \text{otherwise,} \end{cases}
\]
which can be expressed as the 2--D system
\[
\left(\begin{matrix} \dot{x} \\ \dot{y} \end{matrix}\right) =
\left( \begin{matrix} f_x \\ f_y  \end{matrix} \right)
\]
where $f_x = y$ and
\[
f_y =
\begin{cases} -x - \gamma y - \lambda & \text{if $x \geq 0$} \\
-x - \gamma y + \lambda & \text{otherwise}.
\end{cases}
\]
The transition points are located at $x=0$ and therefore
\[
f \cdot \hat{n} = f \cdot \left(\begin{matrix}1 \\ 0
\end{matrix} \right) = f_x = y
\]
which is continuous and therefore has the same sign on both
sides. Therefore there is no mass accumulation at the interface
for this system. This is a special case of theorem~\ref{thm:noaccumulation}.

\subsubsection{Case~3 revisited}

To connect with the example presented in case~3 ($\mu=0$, $\sigma=1$) we choose $w(\lambda)= e^{-\lambda^2/2}$ and $\psi_k$'s as the probabilist's Hermite polynomials $H_k$, with the following properties
\begin{align*}
\int H_i(\lambda) H_k (\lambda) w(\lambda) d\lambda &= k! \sqrt{2\pi} \delta_{ik} \\
H_{k+1}(\lambda) &= \lambda H_k(\lambda) - H_k'(\lambda) \\
H_k'(\lambda) &= k H_{k-1}(\lambda).
\end{align*}
Calculating the terms in Eqn.~\ref{liouville_expanded}:
\begin{align*}
\int w H_i H_k f_x d\lambda &= y \, i! \sqrt{2\pi} \delta_{ik} \\
\int w H_i H_k f_y d\lambda &= - (x+\gamma y) \, i! \sqrt{2\pi}
\delta_{ik} \mp i! \sqrt{2\pi} (\delta_{i,k+1} + k
\delta_{i,k-1}),
\end{align*}
where the upper sign ($-$) is for $x\ge 0$ and the lower sign
($+$) is for~$x < 0$. Substituting into Eqn.~\ref{liouville_expanded} we
obtain the equations
\begin{equation}
\begin{split}
0 &= \partial_t a_0 + \partial_x (y a_0) + \partial_y \left[-(x + \gamma y)a_0 \mp a_1 \right]  \\
0 &= \partial_t a_i + \partial_x (y a_i) + \partial_y \left[-(x
+ \gamma y)a_i \mp (i+1) a_{i+1} \mp a_{i-1} \right] \quad (i\ge
1).
\end{split}
\label{oscillator_hyperbolic_system}
\end{equation}
Note that, instead of using the Hermite polynomials, one can use the Haar wavelet expansion~\cite{LeMaitre2004} to represent the solution as was done in previous sections. We plan to present this calculation in future work.
\begin{thm}
\label{thm:oscillatorishyperbolic}
The system of PDEs for the switching oscillator is hyperbolic.
\end{thm}

\begin{proof}

Consider a system of PDEs of the form
\[
\partial_t u + \sum_\nu A_\nu \partial_\nu u = B,
\]
where $u(x_1, \ldots, x_n, t) \in \mathbb{R}^m$ and the $A_\nu$
are $m\times m$ matrices. The system is hyperbolic if for any
$\alpha_\nu \in \mathbb{R}$ the linear combination $A = \sum_\nu
\alpha_\nu A_\nu$ has real eigenvalues.

For the switching oscillator the system of PDEs can be written
as
\begin{multline*}
\partial_t \begin{pmatrix} a_0 \\ a_1 \\ \vdots \end{pmatrix}
+ \begin{pmatrix} y & 0 & \hdotsfor{2} \\ 0 & y &  &  \\ \vdots
& & \ddots & \end{pmatrix}
\partial_x \begin{pmatrix} a_0 \\ a_1 \\ \vdots \end{pmatrix} \\
+ \begin{pmatrix}
\beta & \mp 1 & 0 & \hdotsfor{2} \\
\mp 1 & \beta & \mp 2 \\
0 & \mp 1 & \beta & \mp 3 \\
\vdots & & \ddots & \ddots & \ddots
\end{pmatrix}
\partial_y \begin{pmatrix} a_0 \\ a_1 \\ \vdots \end{pmatrix}
= \gamma \begin{pmatrix} a_0 \\ a_1 \\ \vdots \end{pmatrix},
\end{multline*}
where $\beta = - (x+ \gamma y)$. Thus, any combination of the
matrices $A_\nu$ is going to be of the tridiagonal form
\[
A = \begin{pmatrix}
a & b & 0 & \ldots & \text{\huge{0}} \\
b & a & 2b & 0 & \ldots\\
0  & b & a & 3b \\
\vdots  & 0   & b &  a \\
\text{\huge{0}}  & \ldots  &   &    & \ddots
\end{pmatrix}.
\]
This tridiagonal non-symmetric matrix is similar to a
tridiagonal symmetric matrix with a diagonal similarity matrix:
$S = D A D^{-1}$, where
\[
D = \begin{pmatrix} \sqrt{0!} &   \\  & \sqrt{1!} & & & \text{\huge{0}}  \\  & &
\sqrt{2!} \\ \text{\huge{0}}& & & \sqrt{3!} \\ & & & & \ddots
\end{pmatrix}.
\]
$A$ is similar to $S$, a symmetric and real matrix, and
therefore $A$~has real eigenvalues.

\end{proof}

The issue of hyperbolicity is discussed in more depth in~\cite{Tryoen2010a}. To solve the hyperbolic system from Eqn.~\ref{oscillator_hyperbolic_system} we write it in the form
\[
\partial_t a + y \partial_x a - (x + \gamma y) \partial_y a \mp A \partial_y a = \gamma a,
\]
where $A$ is the tridiagonal matrix
\begin{equation}
A =
\begin{pmatrix}
0 & 1 \\
1 & 0 & 2 & & \text{\huge{0}}\\
  & 1 & 0 & 3 \\
 \text{\huge{0}} &   & 1 & 0 \\
  &   &   &   & \ddots
\end{pmatrix}.
\label{hermite_matrix}
\end{equation}
We diagonalize $A = P \Lambda P^{-1}$ and define $b = P^{-1} a$
to obtain the set of uncoupled hyperbolic PDEs
\begin{equation}
\partial_t b_i + y \partial_x b_i + \left[-(x + \gamma y) \mp \lambda_i\right] \partial_y b_i = \gamma b_i,
\label{diagonalized_system}
\end{equation}
where $\lambda_i$ is the $i$-th eigenvalue. We will now prove that when the expansion
is truncated up to $a_{n-1}$, i.e., $A$ is truncated to a $n\times n$
matrix, the eigenvalues $A$ are the zeros of~$H_n$.

\begin{thm}
\label{thm:hermite_zeros}
The eigenvalues of $A_n$, the $n\times n$ truncated version of
the matrix in Eqn.~\ref{hermite_matrix}, are the zeros of the $n$-th
order probabilist's Hermite polynomial~$H_n$.
\end{thm}

\begin{proof}

We proceed by induction to prove that $\det (A_n - \lambda I) =
H_n(-\lambda)$, which will then, by the symmetry of $H_n$, prove
our result.

Let $B_n = A_n - \lambda I$. Indeed, $\det B_1 = -\lambda$ and
$\det B_2 = \lambda^2 - 1$. For the general case,
\[
B_{n+1} =
\begin{pmatrix}
  &        &   &   & 0 \\
  & B_n    &   &   & \vdots \\
  &        &   &   & 0 \\
  &        &   &   & n \\
0 & \cdots & 0 & 1 & -\lambda
\end{pmatrix}
\]
Therefore,
\[
\det B_{n+1} = -\lambda \det B_n + (-1)^n n (-1)^{n-1} \det
B_{n-1} = - \lambda \det B_n - n \det B_{n-1},
\]
which is the recurrence relation satisfied by $H_n(-\lambda)$.

\end{proof}

The characteristic curves of Eqn.~\ref{diagonalized_system} are
given by
\begin{align*}
\dot{x} &= y \\
\dot{y} &= - x - \gamma y \mp \lambda_i \\
\dot{b}_i &= \gamma b_i.
\end{align*}
In other words, the characteristics are damped oscillators where
the equilibrium position is given by the eigenvalues of~$\mp A$.
The exponential growth of $b_i$ along a trajectory is due to the
contraction in phase space produced by the dissipation~$\gamma$.

\subsubsection{Results}

We now show results obtained by using transport theory approach on case~3 ($\mu=0$, $\sigma=1$)  for the switching oscillator (Eqn.~\ref{eq:SHMswitch}).

In Fig.~\ref{Fig:MC_vs_transport_colormaps}, we show a series of probability distribution snapshots for Monte Carlo ($5000$ samples) and the transport operator method (for case~3), gridded in the $(x,y)$ plane. Note that we set $y=\dot x$, as defined in the first part of the paper. The Monte Carlo color map snapshots show that the distribution lies on a one-dimensional manifold (in two dimensional space). The one-dimensional nature of the distribution arises because we chose deterministic initial conditions. As discussed previously, all trajectories in case~3 with $\lambda \geq 0$ converge to the origin, resulting in a jump in the cumulative distribution function (CDF).

\begin{figure}%
\centerline{%
\includegraphics[width=\columnwidth]{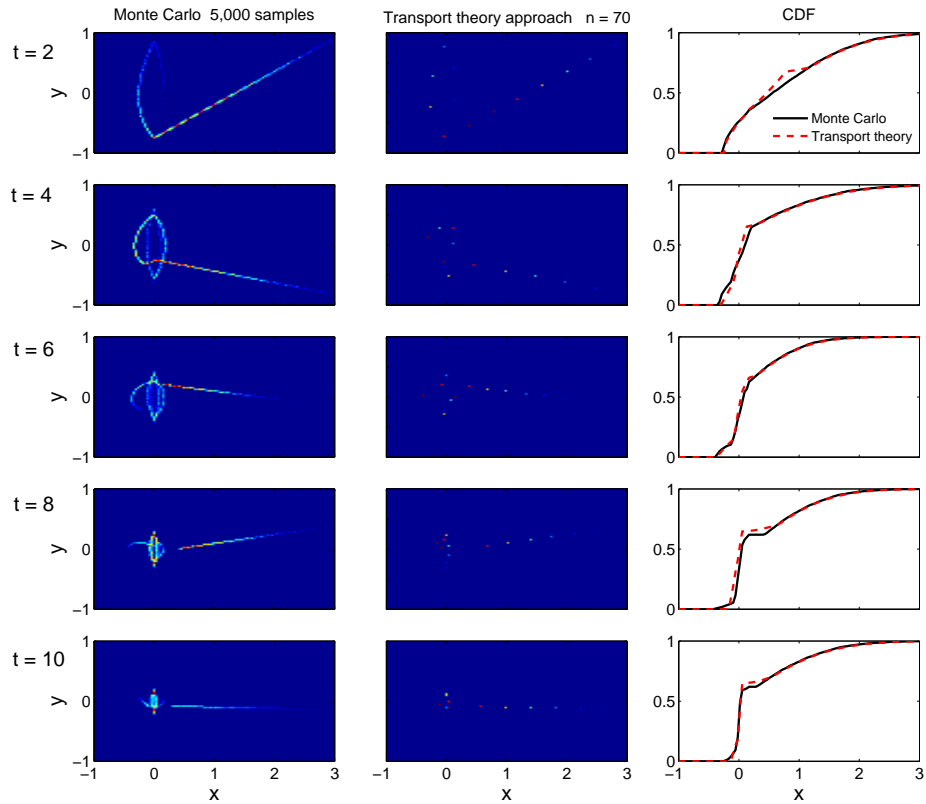}}%
\caption{Snapshots at $t=2,4,6,8,10$ of the gridding from 5000 samples of Monte Carlo and an order 70 expansion using the transport theory based method. The right column compares the one-dimensional cumulative distribution function for the corresponding times.}%
\label{Fig:MC_vs_transport_colormaps}%
\end{figure}

\begin{figure}%
\centerline{%
\includegraphics[width=0.7\columnwidth]{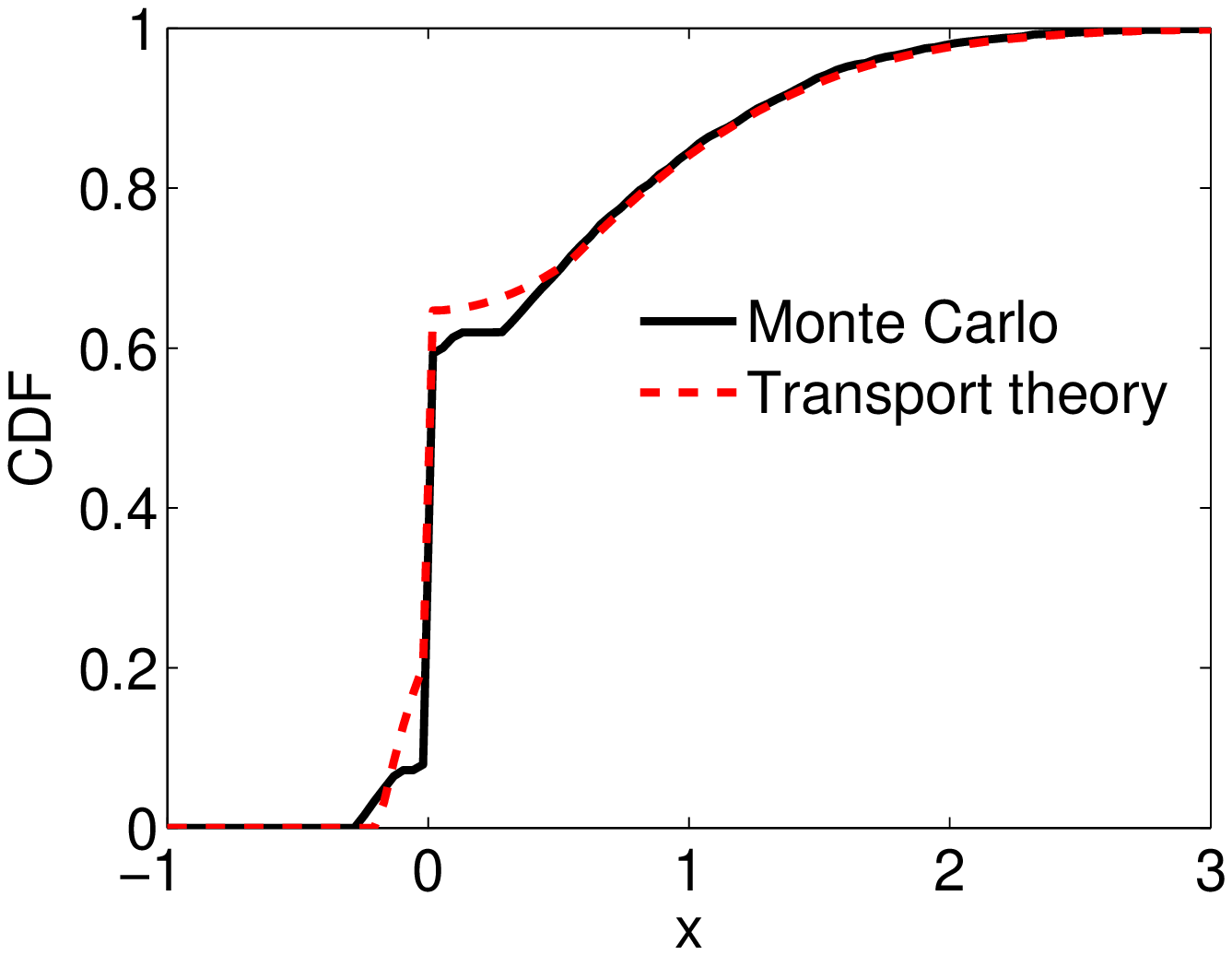}}%
\caption{Comparison at $t=18$ of the cumulative distribution function from 5000 samples of Monte Carlo (solid) and an order~70 expansion using the transport theory based method (dashed).}%
\label{Fig:MC_vs_transport_CDF}%
\end{figure}

For $\lambda < 0$, however, the trajectories converge asymptotically to either $+\lambda$ or $-\lambda$. This convergence to $+\lambda$ or $-\lambda$ is highly dependent on the individual trajectory and results in a fragmentation of the output distribution. Close to convergence ($t=18$), very few trajectories converge to a point in the range $0.2 < x < 0.4$, as seen in the flat region of the CDF in Fig.~\ref{Fig:MC_vs_transport_CDF}. The transport-based method captures the singularity at the origin accurately, but is unable to accurately capture the fragmentation. This is because the method samples the distribution sparsely (determined by the order of expansion), resulting in UQ acceleration. However, this sparsity makes the method miss such fine details. On using a high order of expansion ($n=70$), some samples partially capture the structure around $x=0$. Note that a much lower order expansion accurately captures the jump at the origin and the asymptotic ($x > 1$) shape of the CDF.

As in Fig.~\ref{Fig:case3}, we compare Monte Carlo (5000 samples) with the transport theory approach (orders~$15$ and~$70$ expansion) in Fig.~\ref{Fig:pde_approach_mean_and_var_case3}. The $\mbox{L}_{\infty}$ (maximum) errors for $n=15$ are $9.56\times 10^{-2}$ (mean) and $7.95\times 10^{-2}$ (variance). For $n=70$ we get $\mbox{L}_{\infty}$ errors of $5.28\times 10^{-2}$ and $4.92\times 10^{-2}$ for mean and variance respectively. As shown in Fig.~\ref{Fig:case3}, the method performs reasonably well, however, the results are not nearly as good as those obtained using hybrid polynomial chaos with the Wiener-Haar wavelet expansion in section~\ref{Sec:hpc}. However, with a better choice of basis functions, one does expect better results. The transport operator theory is attractive as it appears to be more versatile. In general, the transport operator approach is applicable to hybrid dynamical systems with overlapping modes of operation (by constructing multiple PDEs for the overlapping mode). In contrast, the hybrid polynomial chaos method suffers from the disadvantage of being inapplicable to such systems.
\begin{figure}
  \subfigure[]{\includegraphics[scale=0.4]{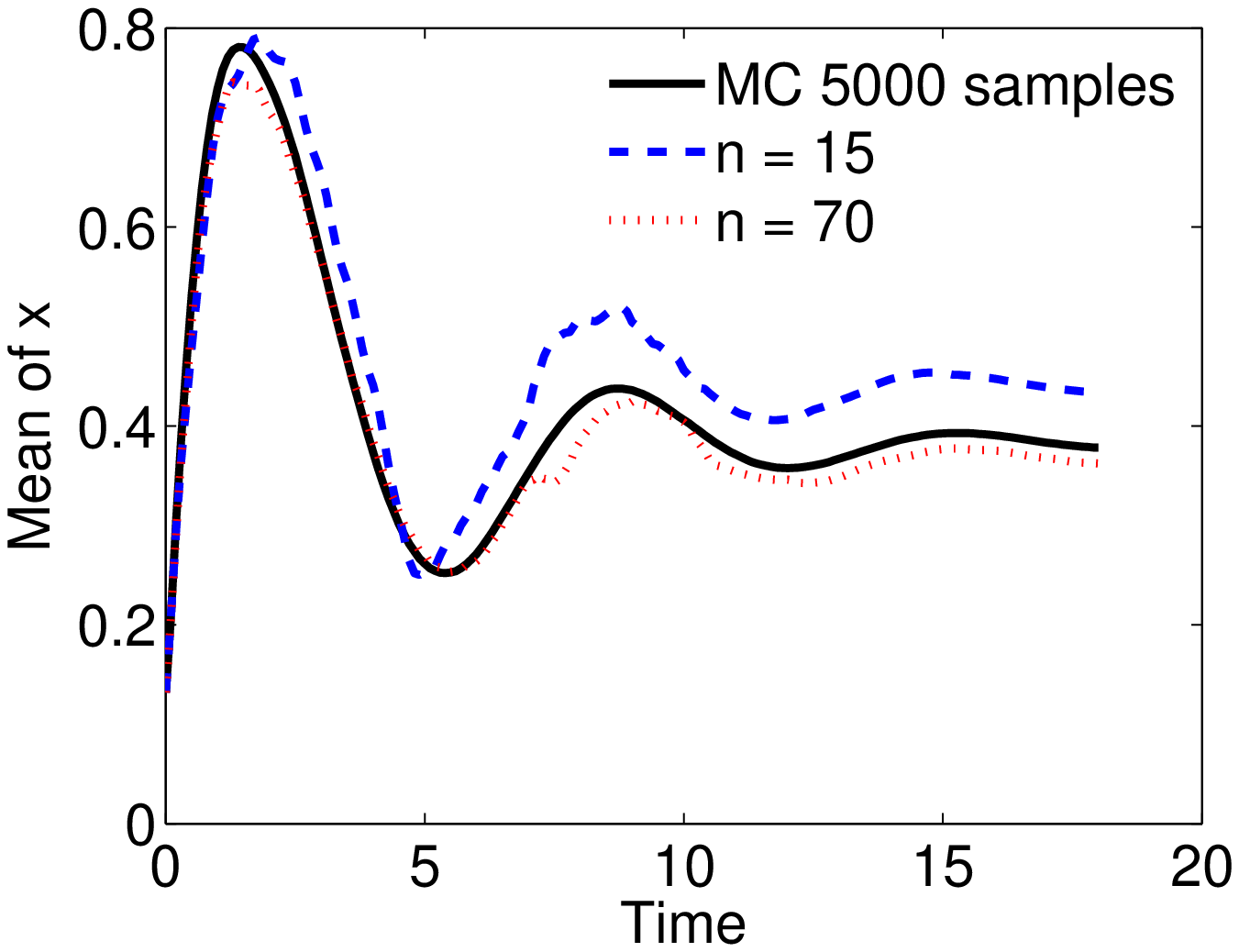}}
  \subfigure[]{\includegraphics[scale=0.4]{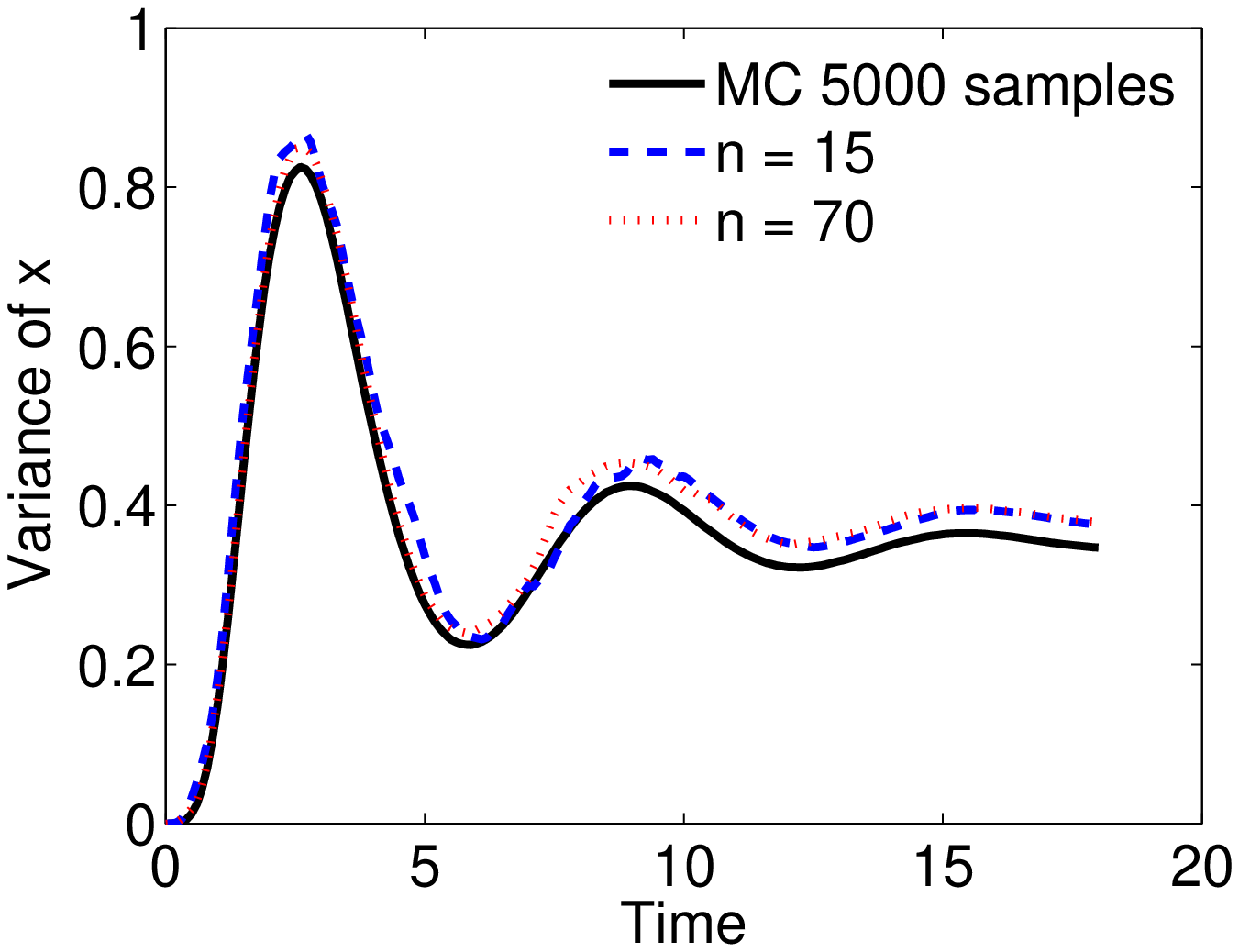}}
  \caption{Comparison of a) predicted mean and b) predicted variance of $x(t;\lambda)$ for case~3. The curves show a 5000-sample Monte Carlo run, and 15 \& 70 term expansions using the transport theory approach. \label{Fig:pde_approach_mean_and_var_case3}}
\end{figure}

\section{Conclusions}\label{Sec:conclusions}
As the modeling of hybrid dynamical systems becomes increasingly
important for modern day engineering applications such as
electrical and biological networks, air traffic systems,
communication networks, etc., quantifying uncertainty in these
systems is going to become a central concern. Since uncertainty
quantification allows one to compute moments of output
distributions in the presence of parametric
uncertainty, these techniques will be used to aid decisions related to robust system design and performance.

In this work, we have made the first attempts to develop fast
uncertainty quantification methods targeted for hybrid dynamical
systems. In particular, we extended polynomial chaos methods, a
popular technique for propagating uncertainty through smooth systems, to hybrid dynamical systems. We also
developed methods to handle state resets within the polynomial
chaos framework by using boundary layer approximations. We then
applied this new approach to perform uncertainty quantification
on switching harmonic oscillators and the bouncing ball examples. We also
demonstrated the efficacy of using Wiener-Haar expansions~\cite{Najm2009, LeMaitre2004,LeMaitre2004a} with our
hybrid polynomial chaos approach for quantifying uncertainty in hybrid systems that
give rise to multi-modal distributions or become increasingly
oscillatory in time. Finally, we showed how a transport theory based approach can capture naturally-emerging discontinuities in the distribution. Future efforts involve providing rigorous error bounds for Wiener-Haar expansions in the hybrid polynomial chaos setting with boundary layer expansions. We are also extending our hybrid polynomial chaos approach to networks of hybrid dynamical systems using our recent work on propagating uncertainty through complex networks~\cite{Tuhin_iter}. We also intend to extend the transport operator based UQ method to hybrid systems with overlapping modes of operation.

\section{Acknowledgements}
The authors thank Habib Najm for pointing us to his work on Wiener-Haar based polynomial chaos expansion and his insightful input. We also thank Alessandro Pinto and George Mathew for valuable discussions related to hybrid dynamical systems.

\bibliographystyle{unsrt}
\bibliography{UQHybridPaper}
\end{document}